\documentclass[11pt]{article}

\usepackage{amsmath,amssymb,amsthm,amsfonts,natbib,xcolor,float,graphicx}
\usepackage[left=2.4cm,top=2.8cm,right=2.4cm,bottom=2.5cm,head=3cm,headsep=1.3cm, foot=1.8cm]{geometry}
\usepackage[bottom,flushmargin,hang,multiple]{footmisc}
\usepackage{subfig}
\usepackage{xr-hyper}
\usepackage{hyperref}
\definecolor{darkblue}{rgb}{0.1,0.1,0.7}
\definecolor{darkred}{rgb}{0.9,0.1,0.1}

\hypersetup{colorlinks, allcolors= darkblue}

\title{Maximum Spectral Measures of Risk\vspace{0.2cm}\\with given Risk Factor Marginal Distributions}
\newtheorem{proposition}{Proposition} 
\newtheorem{lemma}{Lemma} 
\newtheorem{theorem}{Theorem} 
\newtheorem{definition}{Definition}
\newtheorem{remark}{Remark}
\newtheorem{corollary}{Corollary}
\newtheorem{example}{Example}

\newcommand{\ve}{\varepsilon}
\newcommand{\cR}{\mathcal{R}}
\newcommand{\cX}{\mathcal{X}}
\newcommand{\cL}{\mathcal{L}}
\newcommand{\cF}{\mathcal{F}}
\newcommand{\cY}{\mathcal{Y}}
\newcommand{\cW}{\mathcal{W}}
\newcommand{\cP}{\mathcal{P}}

\newcommand{\cM}{\mathcal{M}}
\newcommand{\cV}{\mathcal{V}}

\newcommand{\bbR}{\mathbb{R}}

\renewcommand{\(}{\left(}
\renewcommand{\)}{\right)}
\renewcommand{\tilde}{\widetilde}

\DeclareMathOperator*{\argmin}{argmin}

\defcitealias{BaselMarketRiskCharge}{BCBS [2019]}

\author{Mario Ghossoub\thanks{Department of Statistics and Actuarial Science, University of Waterloo, {\color{blue}mario.ghossoub@uwaterloo.ca}}
\and Jesse Hall \thanks{Department of Statistics and Actuarial Science, University of Waterloo, {\color{blue}jshall@uwaterloo.ca}}
\and David Saunders \thanks{Corresponding author. Department of Statistics and Actuarial Science, University of Waterloo, {\color{blue}dsaunders@uwaterloo.ca}\vspace{0.2cm}}}

\begin{document}
\maketitle

\begin{abstract}
We consider the problem of determining an upper bound for the value of a spectral risk measure of a loss that is a general nonlinear function of two factors whose marginal distributions are known, but whose joint distribution is unknown. The factors may take values in complete separable metric spaces. We introduce the notion of Maximum Spectral Measure (MSP), as a worst-case spectral risk measure of the loss with respect to the dependence between the factors. The MSP admits a formulation as a solution to an optimization problem that has the same constraint set as the optimal transport problem, but with a more general objective function. We present results analogous to the Kantorovich duality, and we investigate the continuity properties of the optimal value function and optimal solution set with respect to perturbation of the marginal distributions. Additionally, we provide an asymptotic result characterizing the limiting distribution of the optimal value function when the factor distributions are simulated from finite sample spaces. The special case of Expected Shortfall and the resulting Maximum Expected Shortfall is also examined.
\end{abstract}

\vspace{1cm}

\noindent\thanks{\vspace{0.5cm}}

\noindent\thanks{ \vspace{0.5cm} }

\noindent

\makeatletter{\renewcommand*{\@makefnmark}{}\footnotetext{We are grateful to Ruodu Wang for comments and suggestions.}\makeatother}

\newpage

%=================================================
%=================================================
%=================================================

\section{Introduction}

A risk measure is a mapping $R$ from a space $\cR$ of random variables on some probability space $(\Omega,\cF,\pi)$ into real numbers, where $\cR$ is interpreted as a collection of discounted net financial positions. Risk measures are often used by financial entities in their risk management process for determining adequate capital requirements that would act as a safety net against adverse deviations from expectations. This determination may be part of the financial entity's enterprise risk management strategy or may be imposed by the regulator. A coherent risk measure (\citet{ArtznerDelbaenEberHeath,Delbaen2002}) is a risk measure that satisfies the following axioms, deemed desirable for effective regulations and management of risk. 
\begin{enumerate}
\renewcommand{\labelenumi}{\textbf{\theenumi}}
\renewcommand{\theenumi}{R.\arabic{enumi}}
\item (Monotonicity) $R(X) \leq R(Y)$, for all $X,Y \in \cR$ such that $X \leq Y$, $\pi$-a.s. 
\item (Positive Homogeneity) $R(\lambda X) = \lambda R(X)$, for all $X \in \cR$ and all $\lambda\in\mathbb{R}_{+}$.
\item (Cash Invariance) $R(X+c) = R(X) + c$, for all $X \in \cR$ and $c\in\mathbb{R}$.
\item (Subadditivity) $R(X+Y) \leq R(X) + R(Y)$ for all $X,Y \in \cR$.
\end{enumerate}

\vspace{0.2cm}

\noindent The most commonly used coherent risk measure in practice is Expected Shortfall (ES), also known as Conditional Value-at-Risk (CVaR). ES is frequently used in quantitative risk management in the banking and insurance industries, forming the basis for the market risk charge in the Basel Framework~(\citetalias{BaselMarketRiskCharge}). For a continuous loss random variable, ES is simply the expected loss given that losses exceed a prescribed quantile. 

\vspace{0.2cm}

A risk measure $R$ is referred to as law-invariant if $R(X)=R(Y)$ when $X$ and $Y$ have the same distribution, and comonotonic additive if $R(X+Y) = R(X)+R(Y)$ when $X$ and $Y$ are comonotonic\footnote{$X,Y$ are said to be comonotonic if for all $\omega,\omega^{\prime}\in \Omega$, $ [X(\omega) - X(\omega^{\prime})][Y(\omega) - Y(\omega^{\prime})] \geq 0 $.}. If the probability space is atomless, then any coherent and comonotonic additive law-invariant risk measure on $\cL^{\infty}(\pi)$ can be represented as a spectral risk measure (\citet{Kusuoka}, \cite{Shapiro2013}, \citet[Theorem 4.93]{FollmerSchied}) of the form 
 \begin{equation*} 
 R(X) = \int_{0}^{1} \mbox{ES}_{u}(X)\, d\Gamma(u),
 \end{equation*}  
for a probability measure $\Gamma$ on $[0,1]$ (see Section~\ref{ProbForm} below). ES at confidence level $\alpha \in (0,1)$ is the most common example of a spectral risk measure, corresponding to $\Gamma = \delta_{\alpha}$, the Dirac delta measure at $\alpha$.

\smallskip

\subsection*{This Paper's Contribution}
We consider the problem of determining an upper bound on the spectral risk measure of a loss that is a nonlinear function of two factors whose marginal distributions are known, but whose joint distribution is unknown; thereby providing a worst-case risk measure for that loss, with respect to the dependence structure between these factors. The original motivation for the formulation of the problem arose from work in counterparty credit risk, in which the two factors correspond to sets of factors for market risk and credit risk (e.g., \citet{RosenSaundersAlpha} and \citet{RosenSaundersST,RosenSaundersCVA}). Bounding the Credit Valuation Adjustment (CVA) (i.e., the ``price" of counterparty credit risk losses) given known market and credit risk factor distributions assumes the form of an optimal transport problem. \citet{GlassermanYang} studied the bounds produced by this problem, and a version of the problem with an entropic penalty function, as well as the convergence of these problems to the true CVA bounding problem, using simulated data. \citet{MemartoluieSaundersWirjanto} considered in a formal way the problem of bounding ES given the distributions of the market and credit factors, and showed that in the case of finite sample spaces, the problem is equivalent to a linear program. Further information on both the CVA and ES bounding problems is contained in~\citet{MemartoluieThesis}.

\vspace{0.2cm}
In this paper, we consider in a more general setting the problem of bounding any spectral risk measure of a loss $L(X,Y)$ given that the marginal distributions of the factors $X$ and $Y$ are known. The factors may take values in complete separable metric spaces $\cX$ and $\cY$ (e.g., spaces of curves or paths, rather than simply vectors in $\mathbb{R}^{n}$). We introduce the Maximum Spectral Measure (MSP), and the special case thereof of Maximum Expected Shortfall (MES) of the loss $L(X,Y)$, as the maximum value that the spectral risk measure $R\left(L(X,Y)\right)$ (resp.\ the ES of $L$) takes over the collection of all probability measures $\pi$ on $\cX \times \cY$ with prescribed marginals $\mu$ on $\cX$ and $\nu$ on $\cY$ for the risk factors $X$ and $Y$. The resulting optimization problem has the same constraint set as the optimal transport problem, but its objective function is more general. We present results analogous to the Kantorovich (strong) duality, and investigate the continuity properties of the optimal value function and optimal solution set with respect to perturbation of the marginal distributions. An asymptotic result characterizing the limiting distribution of the optimal value function in the special case of ES when the factor distributions are simulated from finite sample spaces is also provided. 

\vspace{0.2cm}

The setting, analysis, and results of this paper can be easily extended to a problem of bounding a spectral risk measure of a loss $L: \cX_1 \times \cdots \times \cX_n \to \mathbb{R}$, where $\cX_i$ is an abstract space and $\mu_i$ is a given marginal distribution on $\cX_i$, for each $i \in \{1, \cdots,n\}$. The resulting problem of finding the worst-case spectral risk measure of $L$ over the collection of all joint distributions $\pi$ on $\prod_{i=1}^n \cX_i$ would have the same constraint set as the multi-marginal optimal transport problem (e.g., \citet{RachevRuschendorf}), but its objective function is more general. Using similar techniques as the ones used in this paper would provide a strong duality theory in that case analogous to the Kantorovich duality in multi-marginal optimal transport. Hereafter, we limit the analysis to the case where $n=2$ and $\cX_1,\cX_2$ are complete separable metric spaces, as this is enough to convey the main ideas behind our results.

\smallskip
 
\subsection*{Related Literature}
This paper contributes to the large literature on the optimal transport problem and its applications (e.g., \citet{RachevRuschendorf} or \citet{VillaniTopicsInOT,VillaniOTOldAndNew}, and the references therein). Applications to economics are discussed in \citet{Galichon,Galichon2017} and to risk measures in~\citet{RuschendorfMathematicalRiskAnalysis}. Optimal transport, and the related martingale optimal transport problem have also been applied to the problem of calculating bounds for prices of financial instruments (e.g., \cite{Galichonetal2014}, \citet{BeiglbockLaborderePenkner}, or \citet{Labordere}, and the references therein).

\vspace{0.2cm}

This paper also contributes to the literature on dependence uncertainty bounds for risk measures of the (aggregate loss) function $L$ (e.g., \citet{McNeilFreyEmbrechts} or \citet{RuschendorfMathematicalRiskAnalysis}, and the references therein). In this literature, one is interested in upper and lower bounds on a risk measure of $L(X_1,\cdots,X_n)$, based on the dependence structure between the individual risk factors $X_i$. Specific assumptions on the shape of the function $L$ are typically made\footnote{For instance, sums or stop-loss functions of several functions (e.g., \citet{McNeilFreyEmbrechts}).}, and the risk factors $X_i$ are typically $\mathbb{R}$-valued. See, for instance, \citet{Embrechtsetal2015} (and references therein) for dependence uncertainty bounds for VaR and ES, in a setting where the loss function is the sum of the individual $\mathbb{R}$-valued risk factors. This paper contributes to this literature by considering general aggregate loss functions $L$, a spectral risk measure for $L$, as well as individual risk factors $X_i$ that are not necessarily $\mathbb{R}$-valued. For the case $n=2$ that we treat explicitly in this paper, both risk factors take values in general complete separable metric spaces. As mentioned above, the analysis can be extended to the case $n \geq 3$.

\vspace{0.2cm}

In practical computations, the distributions of the risk factors will often have to be simulated, and a finite dimensional approximation to the original problem, based on the empirical measures, be considered. Important questions then arise regarding the convergence of the approximating problems, and the distribution of the error in the approximate solution. A wide literature on the analogous problem for optimal transport exists, particularly in the case of power cost functions (empirical Wasserstein distances and their convergence properties). For example, the asymptotic rate of convergence to zero of the distance between a sample from the uniform distribution on the cube and the uniform distribution itself is dimension dependent (e.g., \citet{AjtaiKomlosTusnady}, \citet{AmbrosioStraTrevisan}, or \citet{TalagrandULB}, and the references therein). Results for Gaussian samples have recently been studied by~\citet{TalagrandGaussian}, \citet{LedouxOne,LedouxTwo}, and \citet{LedouxThree}. Quantitative estimates on empirical Wasserstein distances are presented in \citet{FournierGuillin}. A central limit type result on the distribution of the error term when the underlying sample spaces are finite was presented in \citet{SommerfeldMunk}, using the approach based on Hadamard differentiability and the delta method that we follow below (see also \citet{KlattTamelingMunk} and \citet{TamelingSommerfeldMunk}). For a central limit type result for a quadratic cost in a more general setting, see \citet{DelBarrioLoubes}, as well as \citet{DelBarrioGordalizaLoubes} for a more general cost in the one-dimensional setting. Comprehensive studies of the case of one space dimension are given in \citet{DelBarrioGineMatran} and \citet{BobkovLedoux}. Estimation of risk measures based on finite samples has also been studied extensively. See, for example \citet{PflugWozabal}, \citet{BelomestnyKratschmer}, \citet{BeutnerZahle}, \citet{KratschmerSchiedZahleTail}, \citet{KratschmerSchiedZahleRobust}, and \citet{KratschmerSchiedZahleQuasiHad}. The closely related topic of the sensitivity of risk measures to changes in the underlying probability measure was investigated by \citet{PichlerThesis,PichlerDifferentProbMeasures}.

\vspace{0.2cm}

The remainder of this paper is structured as follows. Section~\ref{ProbForm} presents a mathematical formulation of the problem of determining the upper bound for a spectral risk measure of a possibly nonlinear function of $X,Y$ given the specification of their marginal distributions, thereby introducing the notion of a MSP. Section~\ref{DualitySpect} presents results analogous to the Kantorovich duality of optimal transport. Section~\ref{StabilitySpect} discusses continuity properties of the optimal value and optimal solution mapping with respect to perturbations of the marginal distributions of the factors $X$ and $Y$. Section~\ref{ESsection} discusses the special case of ES and the resulting MES in further detail. Section \ref{AsymptoticDistributionSection} presents a result on the asymptotic distribution of the optimal value of the ES bounding problem when the distributions of $X$ and $Y$ are simulated from finite sample spaces, and also presents some numerical results for when continuous random variables are simulated. Section~\ref{Conclusion} concludes and suggests several directions for future research.

\smallskip
%=================================================
%=================================================
%=================================================

\section{Spectral Risk Measures and Problem Formulation}
\label{ProbForm}

\subsection{Notation}
Throughout, we assume that $\cX$ and $\cY$ are two Polish (i.e., complete, separable, metric) spaces, and we set $\cW := \cX \times \cY$. When defining measures and random variables, all the spaces are equipped with their respective Borel sigma-algebras. For a Polish space $\cV$, $v\in \cV$, and $\delta > 0$, we denote by $B_{\delta}(v)$ the open ball of radius $\delta$ with centre $v$. Let $\cM(\cV)$  (resp., $\cM_{+}(\cV)$) be the set of finite (resp.\ finite and nonnegative) measures on $\cV$, and let $\cP(\cV)$ be the set of all probability measures on $\cV$. 

\bigskip

Consider a measurable space $\(\Omega, \mathcal{G}\)$. We are given random variables $X: \Omega \to \cX$ and $Y: \Omega \to \cY$, with respective laws $\mu\in\cP(\cX)$ and $\nu\in\cP(\cY)$. We interpret $X$ and $Y$ as exogenously given risk factors, or sets of risk factors. Let $\Pi(\mu,\nu)$ denote the set of all $\pi\in \cP(\cX \times \cY)=\cP(\cW)$ with marginal distributions $\mu$ and $\nu$. For sequences $\{\pi_{n}\}_{n \in \mathbb{N}}\subset\cP(\cV)$, $\pi_{n}\to\pi$ denotes weak convergence (narrow convergence) of probability measures. That is,
$$\pi_{n}\to\pi \ \ \Longleftrightarrow \ \ \int f d\pi_{n}\to \int f d\pi, \ \forall f\in C_{b}(\cV),$$
where $C_{b}(\cV)$ denotes the space of bounded continuous functions from $\cV$ to $\mathbb{R}$, which we note is metrizable using (for example) the Prokhorov metric\footnote{See, e.g., \citet[Chap.\ 11]{DudleyRAP} or \citet[Chap.\ 3]{EthierKurtz}.} on $\cP(\cV)$ defined by 
\begin{equation*} 
d_{\cP}(P,Q) := \inf_{m\in\Pi(P,Q)} \inf\{ \ve > 0: \; m[(x,y): d_{\cV}(x,y) \geq \ve ] \leq \ve\}.
\end{equation*} 

\bigskip

For a loss random variable $L:\cW\to\mathbb{R}$ and a probability measure $\pi\in\cP(\cW)$, we define the generalized inverse cumulative distribution function of $L$ (with respect to $\pi$) as
 \begin{equation*} 
F^{\leftarrow}_{L,\pi}(\alpha) := \inf\{ x\in \mathbb{R} : \pi(L\leq x) \geq \alpha\}, \ \forall \alpha \in (0,1).
\end{equation*} 
$F^{\leftarrow}_{L,\pi}(\alpha)$ is often also called the 
Value-at-Risk of $L$ (under the measure $\pi$) at the confidence level $\alpha$, and denoted by $\mbox{VaR}_{\alpha,\pi}(L)$.

\smallskip
%=================================================
\subsection{Spectral Risk Measures}
A spectral risk measure of the loss random variable $L:\cW\to\mathbb{R}$ is defined as
\begin{equation*} \label{SpectralRiskMeasureVaRDef}
R_{\pi,\sigma}(L) := \int_{0}^{1} F^{\leftarrow}_{L,\pi}(u) \sigma(u)\, du = \int_{0}^{1}\mbox{VaR}_{L,\pi}(u)\sigma(u)\, du,
\end{equation*}
where the spectral function $\sigma: [0,1)\to \mathbb{R}_{+}$ is nonnegative, increasing, and satisfies $\int_{0}^{1}\sigma(u)\, du = 1$. We will also assume that $\sigma$ is right-continuous. For an atomless probability space $(\Omega,\cF,\pi)$, any coherent, comonotonic additive, and law-invariant risk measure on $\cL^{\infty}(\pi)$ can be represented as a spectral risk measure (e.g., \citet{Kusuoka}, \cite{Shapiro2013}, \citet[Theorem 4.93]{FollmerSchied}). 

\bigskip

\noindent
{\bf We assume throughout that we are given a fixed and {\em bounded} spectral function $\sigma$.}

\bigskip
A particular spectral risk measure of significant importance in practice (e.g., \citet{McNeilFreyEmbrechts}, \citet{BaselMarketRiskCharge}) is Expected Shortfall (ES), defined by $\sigma(u) := (1-\alpha)^{-1}\mathbf{1}_{[\alpha,1]}(u)$, that is,
\begin{equation*} 
\mbox{ES}_{\alpha,\pi}(L) := (1-\alpha)^{-1}\int_{\alpha}^{1} F^{\leftarrow}_{L,\pi} (u) \, du.
\end{equation*} 
An important characterization of ES is due to~\citet{RockafellarUryasev} (see also~\citet[Proposition 4.51]{FollmerSchied}). Given $\pi\in\cP(\cW)$ and $\alpha \in (0,1)$, define
\begin{equation} \label{gDefinition}
g_{\alpha}(b,\pi) := b + (1-\alpha)^{-1}\mathbb{E}_{\pi}[(L-b)_{+}].
\end{equation} 
Then
\begin{equation} 
\mbox{ES}_{\alpha,\pi}(L) = \min_{b\in\mathbb{R}} \, g_{\alpha}(b,\pi).
\label{CompactESDef}
\end{equation} 
If instead $\alpha=0$, then 
\begin{equation*} 
\mbox{ES}_{0,\pi}(L) = \mathbb{E}_{\pi}[L] = \inf_{b\in\mathbb{R}} \, g_{0}(b,\pi) = \lim_{b\to -\infty} \, g_{0}(b,\pi),
\label{ESZeroAlpha}
\end{equation*} 
but the infimum may not be attained.

\bigskip
Defining $\gamma_{\sigma}$ to be the measure with distribution function $\sigma$,\footnote{That is, $\gamma_{\sigma}((a,b]) = 
\sigma(b)-\sigma(a)$. Note that $\gamma_{\sigma}$, although finite, is in general not a probability measure.} and $\Gamma_{\sigma}$ to be the probability measure with distribution function $(1-u)\sigma(u)+\int_{0}^{u}\sigma(v)\, dv$, a simple calculation shows that the spectral risk measure $R_{\pi,\sigma}$ can also be expressed as (e.g., \citet[Proposition 8.18]{McNeilFreyEmbrechts})
\begin{equation*} 
R_{\pi,\sigma}(L) = \int_{0}^{1} \mbox{ES}_{u,\pi}(L)\, d\Gamma_{\sigma}(u),
\end{equation*} 
so that $R_{\pi,\sigma}(L)$ can be expressed either as an average of VaRs or of ESs. A simple calculation using integration by parts shows that $d\Gamma_{\sigma}(u) = (1-u) d\gamma_{\sigma}(u)$. 

\bigskip
A natural domain for $R_{\pi,\sigma}$ is studied in~\citet{PichlerBanachSpaces}. Since we assume that $\sigma$ is bounded, this domain is $\cL^{1}(\pi)$ and we have that 
\begin{equation*} 
\|L\|_{1} \leq R_{\pi,\sigma}(|L|) \leq \|\sigma\|_{\infty} \cdot \|L\|_{1}
\end{equation*} 
(e.g., \citet[Theorem 11]{PichlerBanachSpaces}). In particular, the risk measure $R_{\pi,\sigma}$ is finite for all $L\in \cL^{1}(\pi)$. Furthermore, since $\tfrac{d\gamma_{\sigma}}{d\Gamma_{\sigma}} \geq 1$, we have that 
$C_{b}([0,1])\subseteq \cL^{1}(\gamma_{\sigma})\subseteq \cL^{1}(\Gamma_{\sigma})$.

\smallskip
%=================================================
\subsection{Maximum Spectral Measures and Dual Problems}

If we interpret $\mu$ and $\nu$ as the distributions of risk factors whose marginal distributions are known, but whose joint distribution is unknown, then the main quantity of interest to us is the maximum value of the spectral risk measure $R_{\sigma,\pi}$ of the loss $L(X,Y)$, among all measures $\pi$ with the prescribed marginals. We call this the Maximum Spectral Measure:

\begin{definition}\label{DefMSP}
For a given loss random variable $L:\cW\to\mathbb{R}$ and given marginal distributions $\mu$ on $\cX$ and $\nu$ on $\cY$, the Maximum Spectral Measure (MSP) associated with $L$ is defined as
\begin{equation} \label{RobustSpectralMeasurePrimal}
\bar R_{\sigma}(L) := \sup_{\pi\in\Pi(\mu,\nu)} R_{\pi,\sigma}(L).
\end{equation} 
\end{definition}

\noindent We interpret $\bar R_{\sigma}(L)$ as a robust spectral risk measure for $L$, in the sense that it provides a worst-case spectral risk measure for $L(X,Y)$ with respect to the dependence structure between the given risk factors $X$ and $Y$ with given respective marginals $\mu$ and $\nu$.

\bigskip
It follows from~\citet[Proposition 4.20]{FollmerSchied} that, as a supremum of suitably well-behaved risk measures, $\bar R_{\sigma}$ is itself a coherent risk measure on the set of bounded functions on $\cW$. In the notation above, we have suppressed the dependence of $\bar R_{\sigma}$ on $\mu$ and $\nu$. When we want to emphasize this dependence, we will write $V_{\sigma}(\mu,\nu) = \bar R_{\sigma}(L)$ (instead suppressing the dependence on $L$). The primal problem $V_{\sigma}(\mu,\nu)$ bears a resemblance to an optimal transport problem (it has the same feasible set). Indeed, it reduces to an optimal transport problem when $\sigma\equiv 1$. Using our notation, this problem is $\sup_{\pi\in\Pi(\mu,\nu)}\mathbb{E}_{\pi}[L(X,Y)]$, and under technical assumptions (e.g., \citet{VillaniOTOldAndNew}), the well-known Kantorovich duality holds: 
\begin{equation} \label{OTDuality}
\max_{\pi\in \Pi(\mu,\nu)}\mathbb{E}_{\pi}[L(X,Y)] = \min_{(\varphi,\psi)\in \Phi_{L}} \int \varphi\, d\mu + \int \psi \, d\nu,
\end{equation} 
where $\Phi_{L}$ is the set of all $(\varphi,\psi) \in \cL^{1}(\mu)\times \cL^{1}(\nu)$ satisfying $\varphi(x) + \psi(y) \geq L$ for $\mu$-a.e.\ $x\in \cX$ and $\nu$-a.e.\ $y\in \cY$.

\bigskip
We aim to prove a result analogous to the Kantorovich duality for Problem~\eqref{RobustSpectralMeasurePrimal}. To define a natural dual problem to the maximization problem $V_{\sigma}(\mu,\nu)$, we begin by proceeding informally. Using the representation of a spectral risk measure as an average of expected shortfalls, together with~\eqref{CompactESDef}, one would obtain 
\begin{align*} 
V_{\sigma}(\mu,\nu) &= \sup_{\pi\in\Pi(\mu,\nu)} \int_{0}^{1} \mathrm{ES}_{u,\pi}(L)\, d\Gamma_{\sigma}(u) \\
&= \sup_{\pi\in\Pi(\mu,\nu)} \int_{0}^{1} \inf_{b\in\mathbb{R}} b + (1-u)^{-1}\mathbb{E}_{\pi}[(L-b)_{+}]\, d\Gamma_{\sigma}(u) \\
&= \inf_{\beta} \int_{0}^{1} \beta(u)\, d\Gamma_{\sigma}(u) + \sup_{\pi\in\Pi(\mu,\nu)}  \mathbb{E_{\pi}}[C^{\beta}(X,Y)] 
\\
&= \inf_{\beta,\varphi(x)+\psi(y)\geq C^{\beta}(x,y)} \int_{0}^{1} \beta(u)\, d\Gamma_{\sigma}(u) + \int \varphi\, d\mu + \int \psi d\nu, 
\end{align*}
where we have defined, for a given measurable $\beta$
\begin{equation} \label{CBetaDef}
C^{\beta}(x,y) :=
\int_{0}^{1} (1-u)^{-1} \max(L(x,y)-\beta(u),0) \,d\Gamma_{\sigma}(u) = \int_{0}^{1} \max(L(x,y)-\beta(u),0)\, d\gamma_{\sigma}(u),
\end{equation} 
and we have informally assumed that we can switch the order of many pairs of operations (infimum and integral, sup and inf, repeated integrals), and that the Kantorovich duality~\eqref{OTDuality}, with $L$ replaced by $C^{\beta}$, can be applied in the final line. The goal of the next section is to make the above calculation rigorous. To that end, we let $\Phi_{C^{\beta}}$ be the set of all $(\varphi,\psi) \in \cL^{1}(\mu)\times \cL^{1}(\nu)$ satisfying $\varphi(x) + \psi(y) \geq C^{\beta}(x,y)$ for $\mu$-a.e.\ $x\in \cX$ and $\nu$-a.e.\ $y\in \cY$.

\bigskip
The dual to Problem~\eqref{RobustSpectralMeasurePrimal} is the following: 
\begin{equation*} 
D_{\sigma}(\mu,\nu) = 
\inf_{\beta\in \cL^{1}(\gamma_{\sigma}), (\varphi,\psi)\in\Phi_{C^{\beta}}} \int_{0}^{1} \beta(u)\, d\Gamma_{\sigma}(u) + \int \varphi\, d\mu + \int \psi\, d\nu.
\end{equation*} 

\smallskip
\noindent We employ the following assumptions.
\smallskip\textit{
\begin{enumerate}
\renewcommand{\labelenumi}{\textbf{\theenumi}}
\renewcommand{\theenumi}{A.\arabic{enumi}}
\item $L$ is upper semicontinuous. 
\item There exist lower semicontinuous $A \in \cL^{1}(\mu)$ and $B\in \cL^{1}(\nu)$ such that 
$L(x,y) \leq A(x) + B(y)$ for $\mu$-a.e.\ $x\in\cX$ and $\nu$-a.e.\ $y\in\cY$. 
\item There exist upper semicontinuous $a \in \cL^{1}(\mu)$ and $b\in \cL^{1}(\nu)$ such that 
$L(x,y) \geq a(x) + b(y)$ for $\mu$-a.e.\ $x\in\cX$ and $\nu$-a.e.\ $y\in\cY$.
\item  $\sigma$ is bounded and right-continuous.
\end{enumerate}
}

\smallskip
%=================================================
%=================================================
%=================================================

\section{Mathematical Analysis and Duality for the MSP Problem}
\label{DualitySpect}

\subsection{Preliminary Results and Weak Duality}
This section presents some preliminary results that allow us to make basic conclusions about $V_{\sigma}(\mu,\nu)$ and $D_{\sigma}(\mu,\nu)$. In particular, we derive the weak duality result that $V_{\sigma}(\mu,\nu) \leq D_{\sigma}(\mu,\nu)$. The following lemma allows us to conclude that the feasible set of the dual problem is nonempty (and thus derive weak duality), as well as providing a collection of functions $\beta$ such that standard assumptions for the Kantorovich duality apply for $C^{\beta}$ as in~\eqref{CBetaDef}.

\smallskip
\begin{lemma} \label{CBetaLemma}
Suppose that $\beta(u) \in \cL^{1}(\gamma_{\sigma})$. Then $C^{\beta}(x,y)$ is upper semi-continuous and finite-valued, and $0\leq C^{\beta}(x,y) \leq A^{\beta}(x) + B^{\beta}(y)$ with $A^{\beta}\in \cL^{1}(\mu)$, $B^{\beta}\in \cL^{1}(\nu)$.
\end{lemma}

\begin{proof}
The lower bound is immediate. Also $\max(L(x,y)-\beta(u),0) \leq |L(x,y)| + |\beta(u)|$. Let $(x_{n},y_{n}) \to (x,y)$. By upper semicontinuity of $L$, $\limsup_{n} L(x_{n},y_{n}) \leq L(x,y)$, and in particular, for large enough $n$, $\max(L(x_{n},y_{n})-\beta(u),0) \leq |L(x,y)|+|\beta(u)|+1 \in \cL^{1}(\gamma_{\sigma})$. Thus we have
\begin{align*} 
\limsup_{n} C^{\beta}(x_{n},y_{n}) &= \limsup_{n} \int_{0}^{1}  \max(L(x_{n},y_{n})-\beta(u),0) \,d\gamma_{\sigma}(u) \\
&\leq \int_{0}^{1}\limsup_{n}\max(L(x_{n},y_{n})-\beta(u),0) \,d\gamma_{\sigma}(u) \\
& \leq \int_{0}^{1} \max(L(x,y)-\beta(u),0)\, d\gamma_{\sigma}(u),
\end{align*} 
by Fatou's Lemma in the second line, and the upper semicontinuity of $L$ in the third line. The assumptions on $L$ imply that
\begin{equation*} 
0\leq C^{\beta}(x,y) \leq \sigma(1)(\max(|A(x)|,|a(x)|) + \max(|B(y)|,|b(y)|)) + \int_{0}^{1} |\beta(u)|\, d\gamma_{\sigma}(u),
\end{equation*} 
from which the result follows.
\end{proof}

\smallskip
The above lemma implies that $\Phi_{C^{\beta}} \neq \emptyset$ when $(1-u)^{-1}\beta(u)$ is integrable with respect to $\Gamma_{\sigma}$, i.e., when $\beta\in \cL^{1}(\gamma_{\sigma})$. Consequently, the feasible set of the dual problem is nonempty. We immediately obtain the following weak duality result.

\smallskip
\begin{proposition}[Weak Duality]
$D_{\sigma}(\mu,\nu) \geq V_{\sigma}(\mu,\nu)$.
\end{proposition}

\begin{proof}
Let $\pi$ be feasible for the primal problem and $(\beta,\varphi,\psi)$ be feasible for the dual problem. Then
\begin{align*} 
\int \beta\, d\Gamma_{\sigma} + \int \varphi\, d\mu + \int \psi\, d\nu & = \int\beta d\Gamma_{\sigma } 
+ \int (\varphi(x) + \psi(y)) \, d\pi(x,y) \\
& \geq \int \int_{0}^{1} \beta(u) + (1-u)^{-1} \max(L(x,y)-\beta(u),0) \, d\Gamma_{\sigma}(u) d\pi(x,y) \\
&= \int_{0}^{1} g_{u}(\beta(u),\pi) d\Gamma_{\sigma}(u)
%\\ &
= \int_{0}^{1} \mbox{ES}_{u,\pi}(L) d\Gamma_{\sigma}(u) = R_{\pi,\sigma}(L),
\end{align*}
and the result follows.
\end{proof}

\smallskip
%=================================================
\subsection{Auxiliary Optimization Problems and Interchanging Operations}
Before we proceed to showing strong duality, we exam some related optimization problems that arise in the proof. We begin with what is basically a restatement of Theorem 1 in \citet{RockafellarUryasev} in our context. In particular, it yields that $\mbox{VaR}_{\alpha,\pi}(L)$ is a minimizer for the optimization problem~\eqref{CompactESDef} whose optimal value is the ES of $L$ (under $\pi$). 
\smallskip
\begin{lemma}
Let $\alpha\in (0,1)$, $\pi\in\cP(\cW)$, and $L \in \cL^{1}(\pi)$, and consider the function $g_{\alpha}(\cdot,\pi):\bbR\to\bbR$ defined by~\eqref{gDefinition}. Then $g_{\alpha}(\cdot,\pi)$ is convex, $\underset{b\to\pm\infty} \lim \,g_{\alpha}(b,\pi) = \infty$, and
\begin{align*} 
\partial_{1} g_{\alpha}(b,\pi) &= [1-(1-\alpha)^{-1}\pi(L\geq b),1-(1-\alpha)^{-1}\pi(L>b)],\\
\argmin g_{\alpha}(\cdot,\pi) &= \{ b | \pi(L > b) \leq 1-\alpha \leq \pi(L\geq b)\},
\end{align*} 
where we use the notation $\partial_{1} g_{\alpha}(b,\pi)$ for the subdifferential of $g_{\alpha}(\cdot,\pi)$ at $b$.
\end{lemma}

\smallskip
\noindent The proof is straightforward: convexity is immediate; and to calculate the subdifferential, one can take difference quotients and use the Dominated Convergence Theorem.

\smallskip
\begin{remark}
It should be noted that we have not assumed that $\sigma(0)=0$. In particular $\gamma_{\sigma}$ may include a point mass at 0. We denote by $\gamma_{\sigma}(\{0\})=z_{0} \leq 1$. Observe that this also implies that $\Gamma_{\sigma}(\{0\})=z_{0}$. We introduce the notation $\tilde \gamma_{\sigma}$ and $\tilde\Gamma_{\sigma}$ for the parts of the measures that live on $(0,1)$: 
\begin{equation*}  \label{TildeMeasuresDef}
\tilde \gamma_{\sigma} := \gamma_{\sigma} - z_{0}\delta_{0}, \quad \tilde\Gamma_{\sigma} := \Gamma_{\sigma} - z_{0} \delta_{0}.
\end{equation*}
\end{remark}

\smallskip
\begin{lemma} \label{FInverseInL1}
Let $\pi\in\Pi(\mu,\nu)$, then $F_{L,\pi}^{\leftarrow} \in \cL^{1}(\tilde \Gamma_{\sigma})$.
\end{lemma}

\begin{proof} 
Suppose that $L \geq 0$. Then
\begin{equation*} 
\int_{0}^{1} F_{L,\pi}^{\leftarrow}(u)\, d\tilde\Gamma_{\sigma}(u) 
\leq \int_{0}^{1} \mbox{ES}_{u,\pi}(L)\, d\tilde\Gamma_{\sigma(u)} \leq R_{\pi,\sigma}(L) \leq \|\sigma\|_{\infty}\cdot \|L\|_{1},
\end{equation*} 
since $L \in \cL^{1}(\pi)$. If $L \leq 0$, then note that by the subadditivity of $R_{\pi,\sigma}$, $0 = R_{\pi,\sigma}(L-L) \leq R_{\pi,\sigma}(L) + R_{\pi,\sigma}(-L)$. Therefore, 
\begin{equation*}
\int_{0}^{1} |F_{L,\pi}^{\leftarrow}(u)| \, d\tilde \Gamma(u) = -\int_{0}^{1} F_{L,\pi}^{\leftarrow}(u)\, d\tilde\Gamma_{\sigma}(u) 
= z_{0} \mathbb{E}_{\pi}[L] - R_{\pi,\sigma}(L)\leq z_{0}\mathbb{E}_{\pi}[L] + R_{\pi,\sigma}(-L) < \infty.
\end{equation*}
The result then follows since $|F_{\pi,L}^{\leftarrow}| \leq \max(|F_{\pi,|L|}^{\leftarrow}|,|F_{\pi,-|L|}^{\leftarrow}|)$.
\end{proof}

\smallskip

Next, we consider the issue of interchanging the infimum and integral.

\smallskip
\begin{lemma}\label{InterchangeInfAndIntegral}
Suppose that $\pi\in\Pi(\mu,\nu)$, and let $\cL^{0}$ be the set of all Borel measurable functions from $[0,1)\to\mathbb{R}$. Then
\begin{multline*} 
K_{\pi}:= \inf_{\beta\in\cL} \left(\int_{0}^{1} g_{u}(\beta(u),\pi)\, d\Gamma_{\sigma}(u)\right) = \int_{0}^{1} \inf_{b} g_{u}(b,\pi) \, d\Gamma_{\sigma}(u) 
 = \Gamma_{\sigma}(\{0\}) \mathbb{E}_{\pi}[L]+ \int_{0}^{1} \mathrm{ES}_{u,\pi}(L) \, d\tilde\Gamma_{\sigma}(u).
\end{multline*}

\smallskip
\noindent Moreover, one has, for $\beta\in\cL^{0}$, that $\beta$ is a minimizer iff $\beta(u)\in\argmin_{b\in\mathbb{R}}(b+(1-u)^{-1}\mathbb{E}_{\pi}[(L-b)_{+}])$ for 
$\Gamma_{\sigma}$-a.e.\ $u\in [0,1)$. Finally,
\begin{equation*} 
K_{\pi}= \inf_{\beta\in \cL^{\infty}(\Gamma_{\sigma})} \left(\int_{0}^{1} g_{u}(\beta(u),\pi)\, d\Gamma_{\sigma}(u)\right).
\end{equation*}
\end{lemma}

\smallskip
\begin{proof} 
Since $\pi$ is fixed, define $f(b,u) = g_{u}(b,\pi)$. All the statements except for the final one follow from \citet[Theorem 14.60]{RockafellarWets} (the fact that $f$ is a normal integrand follows from \citet[Example 14.31]{RockafellarWets}). For the last statement, take $\beta_{n}(0) = -n$, $\beta(u)= F_{L,\pi}^{\leftarrow}(u)$, and $\beta_{n}(u) = \beta(u) \cdot \mathbf{1}_{\{ -n \leq \beta(u) \leq n\}}$ for $u\in (0,1)$. Then $\beta_{n} \in \cL^{\infty}(\Gamma_{\sigma})$. For $u=0$, $g_{0}(\beta_{n}(0),\pi) \to \mathbb{E}_{\pi}[L]$, and for $u\in (0,1)$, $|g_{u}(\beta_{n}(u),\pi)| \leq \mathrm{ES}_{u,\pi}(|L|)$, so that the Dominated Convergence Theorem implies that $\int_{0}^{1}g_{u}(\beta_{n}(u),\pi)\, d\tilde \Gamma_{\sigma}(u) \to \int_{0}^{1} ES_{u,\pi}(L)\, d\tilde\Gamma_{\sigma}(u)$, and the result follows.
\end{proof}

\smallskip
\begin{remark}
As noted earlier, when $z_{0} > 0$, the minimum is in general not attained. In the case when $z_{0} = 0$, the minimum is attained by $\beta(u) = F_{L,\pi}^{\leftarrow}(u)$, and we have that $\beta \in \cL^{1}(\Gamma_{\sigma})$ by Lemma~\ref{FInverseInL1}. However, it may not be the case that $\beta \in \cL^{1}(\gamma_{\sigma})$, as needed in Lemma~\ref{CBetaLemma} above in order to ensure that $C^{\beta}$ satisfies the hypotheses required to apply standard results on optimal transport duality.  For example, if $\sigma(u) = 3(1-\sqrt{1-u})$, so that $d\gamma_{\sigma} = \tfrac{3}{2}(1-u)^{-1/2}$, and $d\Gamma_{\sigma} = \tfrac{3}{2}(1-u)^{1/2}$, then by taking $L(X,Y)=X$, and assuming that under $\mu$, the real-valued random variable $X$ has distribution function $F_{X}(x) = (1-x^{-\theta})\mathbf{1}_{\{x\geq 1\}}$ for $\theta\in (1,2)$, we obtain a situation satisfying our assumptions, but with $F_{L,\pi}^{\leftarrow} \in \cL^{1}(\Gamma_{\sigma})\setminus \cL^{1}(\gamma_{\sigma})$ for all $\pi\in\Pi(\mu,\nu)$.
\end{remark}

\smallskip
Define the function $H:\Pi(\mu,\nu)\times \cL^{1}(\Gamma_{\sigma})\to {\mathbb{R}}\cup \{\infty\}$ by
\begin{equation*} 
H(\pi,\beta) := \int_{0}^{1} \beta(u) + (1-u)^{-1}\mathbb{E}_{\pi}[(L(X,Y)-\beta(u))_{+}]\, d\Gamma_{\sigma}(u).
\end{equation*} 
The integral of the first term will always be finite by assumption, but we have only shown that the integral of the second term is finite in the case where $\beta\in \cL^{1}(\gamma_{\sigma})\subseteq \cL^{1}(\Gamma_{\sigma})$. Let
\begin{equation*} 
\mathcal{D}_{H} := \{ \beta\in L^{1}(\Gamma_{\sigma}) : \forall \pi\in\Pi(\mu,\nu),\;H(\pi,\beta) < \infty \},
\end{equation*} 
and note that $\cL^{\infty}(\gamma_{\sigma}) \subseteq \cL^{1}(\gamma_{\sigma})\subseteq \mathcal{D}_{H}$.

\bigskip\noindent
We make the following observations. 

\smallskip
\begin{proposition} \label{HPropertiesProposition} \ 
\begin{enumerate} 
\item For a given $\beta\in \mathcal{D}_{H}$, $H(\cdot,\beta)$ is concave in $\pi$. In fact for $\pi_{0},\pi_{1}\in\Pi(\mu,\nu)$ and $\lambda\in (0,1)$: $H(\lambda\pi_{1} + (1-\lambda)\pi_{0},\beta) = \lambda H(\pi_{1},\beta) + (1-\lambda)H(\pi_{0},\beta)$. Furthermore, if $\beta\in \cL^{1}(\gamma_{\sigma})$, $H(\cdot,\beta)$ is upper semicontinuous in $\pi$, i.e., if $\pi_{n} \to \pi$, then $\limsup H(\pi_{n},\beta) \leq H(\pi,\beta)$.
\item For a fixed $\pi\in\Pi(\mu,\nu)$, $H(\pi,\cdot)$ is convex (and therefore $\mathcal{D}_{H}$ is a convex subset of $\cL^{1}(\Gamma_{\sigma})$).
\end{enumerate}
\end{proposition}

\begin{proof}
Linearity in $\pi$ is clear, while the second result follows from the convexity of the function $b\to (L-b)_{+}$. It remains to show upper semicontinuity of $H(\beta,\cdot)$. Suppose that $\beta\in \cL^{1}(\gamma_{\sigma})$ is fixed, then, defining
\begin{equation*} 
g(x,y) := \int (1-u)^{-1} (L(x,y)-\beta(u))_{+}\, d\Gamma_{\sigma}(u),
\end{equation*} 
and using the fact that $0 \leq (L(x,y)-\beta(u))_{+} \leq |\beta(u)| + \max(A(x)+B(y),-a(x)-b(y))$, it follows that $g(x,y)$ is upper semicontinuous and bounded above by $\|\beta\|_{1,\gamma} + \sigma(1)\max(A(x)+B(y),-a(x)-b(y))$. Thus, upper semicontinuity of $H(\cdot,\beta)$ then follows from \citet[Lemma 4.3]{VillaniOTOldAndNew}.
\end{proof}

\smallskip

The following is now a consequence of standard minimax theorems (e.g., \citet[Theorem 2.10.2]{Zalinescu}).

\smallskip
\begin{proposition}\label{MinMaxProposition}
The following holds:
\begin{multline*}
\inf_{\beta\in \cL^{1}(\gamma)}\max_{\pi\in\Pi(\mu,\nu)} \left(\int_{0}^{1} \beta(u) + (1-u)^{-1}\mathbb{E}_{\pi}[(L(X,Y)-\beta(u))_{+}]\, d\Gamma(u)\right) 
\\
= \max_{\pi\in\Pi(\mu,\nu)} \inf_{\beta\in \cL^{1}(\gamma)} \left(\int_{0}^{1} \beta(u) + (1-u)^{-1}\mathbb{E}_{\pi}[(L(X,Y)-\beta(u))_{+}]\, d\Gamma(u)\right). 
\end{multline*}
\end{proposition}

\smallskip
%=================================================
\subsection{Strong Duality}

We are now ready to prove strong duality.

\smallskip
\begin{theorem}[Strong Duality] Given assumptions A.1-A.4 we have
\begin{equation*} 
V_{\sigma}(\mu,\nu) = D_{\sigma}(\mu,\nu).
\end{equation*} 
\end{theorem}

\begin{proof}
\begingroup
\allowdisplaybreaks
\begin{align} 
\inf_{\beta\in \cL^{1}(\gamma_{\sigma}), (\varphi,\psi)\in\Phi_{C^{\beta}}} &\int_{0}^{1} \beta(u)\, d\Gamma_{\sigma}(u) + \int \varphi\, d\mu + \int \psi\, d\nu
\nonumber
\\
&= \inf_{\beta\in \cL^{1}(\gamma_{\sigma})} \left(\int_{0}^{1}\beta(u)\, d\Gamma_{\sigma}(u) + \inf_{(\varphi,\psi)\in \Phi_{C^{\beta}}}\int \varphi\, d\mu + \int \psi\, d\nu \right) 
\nonumber
\\
&= \inf_{\beta\in \cL^{1}(\gamma_{\sigma})} \left( \int_{0}^{1} \beta(u)\, d\Gamma_{\sigma}(u) + \max_{\pi\in\Pi(\mu,\nu)} \mathbb{E}_{\pi}[C^{\beta}(x,y)]\right) 
\label{OTDualityStep}
\\
&= \inf_{\beta\in \cL^{1}(\gamma_{\sigma})} \max_{\pi\in\Pi(\mu,\nu)} \left(\int_{0}^{1} \beta(u)\, d\Gamma_{\sigma}(u) + \int_{\cX\times\cY} \int_{0}^{1} (1-u)^{-1}\max(L(x,y)-\beta(u),0)\, d\Gamma_{\sigma}(u)d\pi \right) \nonumber
\\
&= \inf_{\beta\in \cL^{1}(\gamma_{\sigma})}\max_{\pi\in\Pi(\mu,\nu)} \left(\int_{0}^{1} \beta(u) + (1-u)^{-1}\mathbb{E}_{\pi}[(L(X,Y)-\beta(u))_{+}]\, d\Gamma_{\sigma}(u)\right) 
\label{TonelliStep}
\\
&= \max_{\pi\in\Pi(\mu,\nu)} \inf_{\beta\in \cL^{1}(\gamma_{\sigma})} \left(\int_{0}^{1} \beta(u) + (1-u)^{-1}\mathbb{E}_{\pi}[(L(X,Y)-\beta(u))_{+}]\, d\Gamma_{\sigma}(u)\right) 
\label{MinMaxStep}
\\ 
&= \max_{\pi\in\Pi(\mu,\nu)} \int_{0}^{1} \inf_{b} (b + (1-u)^{-1}\mathbb{E}_{\pi}[(L(X,Y)-b)_{+}]) \, d\Gamma_{\sigma}(u),
\label{InterchangeStep}
\\
&= \max_{\pi\in\Pi(\mu,\nu)} \int_{0}^{1} \mathrm{ES}_{u}(L)\, d\Gamma_{\sigma}(u), \nonumber
\end{align} 
\endgroup

\smallskip
\noindent where~\eqref{OTDualityStep} follows from the Kantorovich Duality for Optimal Transport (e.g., \citet[Theorem 5.10]{VillaniOTOldAndNew}), \eqref{TonelliStep} uses Tonelli's Theorem, \eqref{MinMaxStep} is due to Proposition~\ref{MinMaxProposition}, and \eqref{InterchangeStep} is due to Lemma~\ref{InterchangeInfAndIntegral}.
\end{proof}

\smallskip
\begin{remark} 
Unlike in the Kantorovich duality, we have not shown attainment of the optimal dual solution. Given the importance of attainment (and uniqueness) of optimal dual solutions for understanding error distributions when approximating our problem by using problems obtained by sampling $\mu$ and $\nu$ (e.g., \citet{SommerfeldMunk}, \citet{DelBarrioLoubes}, or Section~\ref{AsymptoticDistributionSection}), this is an interesting direction for future research.
\end{remark}

\smallskip
%=================================================
%=================================================
%=================================================

\section{Stability of the MSP}
\label{StabilitySpect}

We consider continuity properties of the optimal values and optimal solution sets of Problem~\eqref{RobustSpectralMeasurePrimal}, when $\mu_{n}\to \mu$ and $\nu_{n} \to \nu$. We begin by recalling the following definitions (e.g., \citet{OkRealAnalysisWithEconomicApplications}).

\smallskip
\begin{definition} 
For two metric spaces $\cV$, $\cW$, a correspondence $G:\cV \rightrightarrows \cW$ is said to be: 
\begin{itemize} 
\item Upper hemicontinuous if for each $v\in\cV$, and every open $O$ of $\cW$ with $G(v) \subseteq O$, there exists some $\delta > 0$ such that $G(B_{\delta}(x)) 
\subseteq O$.
\item Lower hemicontinuous if for every $v\in\cV$, and every open $O$ with $G(v)\cap O \ne \emptyset$, there exists some $\delta > 0$ such that $G(v')\cap O \ne \emptyset$ for all $v' \in B_{\delta}(v)$.
\item Continuous if it is both upper hemicontinuous and lower hemicontinuous.
\end{itemize} 
\end{definition} 

\smallskip
\noindent Hemicontinuity can also be characterized in terms of sequences. $G$ is upper hemicontinuous if for any sequence $\{ v_{n}\}_{n\in\mathbb{N}}$ in $\cV$ and any sequence $\{w_{n}\}_{n\in\mathbb{N}}$ with $v_{n} \to v$ and $w_{n} \in G(v_{n})$ for each $n$, there exists a subsequence $\{ w_{n_{k}}\}_{k\in\mathbb{N}}$  that converges to a point in $G(v)$. $G$ is lower hemicontinuous if for any sequence $\{v_{n}\}_{n\in\mathbb{N}}$ with $v_{n} \to v$ and any $w\in G(v)$, there exists $\{w_{n}\}_{n\in\mathbb{N}}$ such that $w_{n} \to w$ and  $w_{n}\in G(v_{n})$ for each $n$.

\bigskip
The following result is proved in~\citet{Bergin} by a discretization argument. A shorter and more direct alternative proof is given in~\citet{GhossoubSaunders2020}.

\smallskip
\begin{proposition}\label{OTCorrespondenceContinuity}  The feasible set mapping $\Pi:\cP(\cX)\times\cP(\cY) 
\rightrightarrows \cP(\cX \times \cY) = 
\cP(\cW)$ is 
continuous.
\end{proposition}

\bigskip\noindent
Recall also the Wasserstein metrics (e.g., \citet{VillaniTopicsInOT}) defined for $r\geq 1$ on $\cP(\cW)$ by: 
\begin{equation*} 
W_{r}(\pi_{1},\pi_{2}) = \left(\inf_{m\in\Pi(\pi_{1},\pi_{2})} \int_{\cW\times\cW} d_{\cV}(u,v)^{r}\, dm(u,v)\right)^{1/r},
\end{equation*} 
for probability measures $\pi_{1},\pi_{2} \in \cP_{r}(\cW)$, the collection of all probability measures on $\cW$ with finite $r^{th}$ moment. We note that $W_{r}(\pi_{n},\pi)\to 0$ if and only if $\pi_{n} \to \pi$ and $\{\pi_{n}\}_{n\in\mathbb{N}}$ has uniformly integrable $r^{th}$ moments, i.e.,
\begin{equation*} 
\lim_{R\to\infty} \limsup_{n\to\infty} \int_{d_{\cW}(w_{0},w)\geq R} d_{\cW}(w_{0},w)^{r} \, d\pi_{n}(w) = 0,
\end{equation*} 
for some (and therefore any) $w_{0} \in \cW$ (e.g., \citet[Theorem 7.12]{VillaniTopicsInOT}). We immediately obtain the following result. 

\smallskip
\begin{corollary} 
The feasible map $\Pi$ is continuous as a correspondence from $\cP_{r}(\cX)\times \cP_{r}(\cY)\rightrightarrows \cP_{r}(\cW)$.
\end{corollary}

\bigskip\noindent
Continuity of $\Pi$ can be combined  with known results on continuity of risk measures in order to derive stability results for $V_{\sigma}(\mu,\nu)$.

\bigskip
The following is Corollary~11 of \citet{PichlerDifferentProbMeasures}. 
\begin{proposition} 
Suppose that $L:\cW \to \mathbb{R}$ is H{\"{o}}lder continuous with exponent $q\leq 1$ and constant $C_{q}$, $|L(w)-L(w')|\leq C_{q}\cdot d_{\cW}(w,w')^{q}$, 
and $L$ is in both $\cL^{p}(\pi)$ and $\cL^{p}(\pi')$ for some $p\geq 1$.  
Then 
\begin{equation*} 
|R_{\pi,\sigma}(L) - R_{\pi',\sigma}(L)| \leq C_{q}\cdot W_{r}(\pi,\pi') \cdot \|\sigma\|_{r_{q}}, 
\end{equation*}
where $r_{q} \geq \tfrac{r}{r-q}$, and the norm of $\sigma$ is taken with respect to Lebesgue measure on $[0,1]$.
\end{proposition}

\bigskip
\begin{proposition} 
Suppose that $\pi_{n},\pi \in \cP_{r}(\cW)$, $L_{n}:\cW\to\mathbb{R}$ are H{\"{o}}lder continuous with exponent $q\leq 1$ and constant $C_{q}$ (independent of $n$), $L_{n}$ is in both  $\cL^{1}(\pi_{n})$ and  $\cL^{1}(\pi)$ for $n=1,2,\ldots$, and $L_{n} \to L$ in $\cL^{1}(\pi)$.  If $W_{r}(\pi_{n},\pi)\to 0$, then $R_{\pi_{n},\sigma}(L_{n}) \to R_{\pi,\sigma}(L)$. 
\end{proposition} 

\begin{proof}
\begin{equation*}
\begin{split}
|R_{\pi_{n},\sigma}(L_{n}) - R_{\pi,\sigma}(L)| 
&\leq |R_{\pi_{n},\sigma}(L_{n}) - R_{\pi,\sigma}(L_{n})| + |R_{\pi,\sigma}(L_{n})-R_{\pi,\sigma}(L)| \\
&\leq C_{q}\cdot W_{r}(\pi,\pi_{n}) \cdot \|\sigma\|_{r_{q}} + \|L_{n}-L\|_{\pi,1}\cdot \|\sigma\|_{\infty}.
\end{split}
\end{equation*}
\end{proof}

\begin{proposition} 
Suppose that $W_{r}(\mu_{n},\mu)\to 0$, $W_{r}(\nu_{n},\nu)\to 0$, and $L_{n}:\cW\to\mathbb{R}$ are H{\"{o}}lder continuous with exponent $q\leq 1$ and constant $C_{q}$ (independent of $n$). Suppose that for any $\pi_{n}\in \Pi(\mu_{n},\nu_{n})$, $\pi\in\Pi(\mu,\nu)$, $L_{n}$ is in both  $\cL^{1}(\pi_{n})$ and  $\cL^{1}(\pi)$ for $n=1,2,\ldots$, and $L_{n} \to L$ in $\cL^{1}(\pi)$. Under the assumptions of the previous propositions, $V_{\sigma}(\mu_{n},\nu_{n}) \to V_{\sigma}(\mu,\nu)$. In addition, if $\{\pi_{n}^{*}\}_{n \in \mathbb{N}}$  is a sequence of maximizers for $V_{\sigma}(\mu_{n},\nu_{n})$, then, up to extraction of a subsequence, there exists $\pi\in\Pi(\mu,\nu)$ with $W_{r}(\pi^{*}_{n},\pi^{*}) \to 0$ and $\pi$ is a maximizer for $V_{\sigma}(\mu,\nu)$.
\end{proposition}

\begin{proof} 
The proof follows a standard argument (e.g., \citet[Theorem 8.6.6]{Lucchetti}).  Let $\pi$ attain $V_{\sigma}(\mu,\nu)$. There exists $\pi_{n} \in \Pi(\mu_{n},\nu_{n})$ such that $W_{r}(\pi_{n},\pi) \to 0$. By the above proposition, it follows that
\begin{equation*} 
V_{\sigma}(\mu_{n},\nu_{n}) \geq R_{\sigma,\pi_{n}}(L_{n}) \to R_{\sigma,\pi}(L) = V_{\sigma}(\mu,\nu),
\end{equation*} 
so that $V_{\sigma}(\mu,\nu) \leq \liminf V_{\sigma}(\mu_{n},\nu_{n})$. Let $\pi_{n}^{*}\in \Pi(\mu_{n},\nu_{n})$ attain $V_{\sigma}(\mu_{n},\nu_{n})$. By passing to a subsequence $n_{k}$ if necessary, we can assume that $\lim_{k\to\infty} V_{\sigma}(\mu_{n_{k}},\nu_{n_{k}}) = \limsup V_{\sigma}(\mu_{n},\nu_{n})$. An application of Prokhorov's Theorem yields that $\{\pi^{*}_{n_{k}}\}_{k \in \mathbb{N}}$ is relatively compact. Extracting a convergent subsequence, which we still denote $\pi_{n_{k}}$, we obtain that: 
\begin{equation*} 
\limsup V_{\sigma}(\mu_{n},\nu_{n}) = \lim_{k\to\infty} R_{\sigma,\pi^{*}_{n_{k}}}(L_{n_{k}}) = R_{\sigma,\pi}(L) \leq V_{\sigma}(\mu,\nu) 
\leq \liminf V_{\sigma}(\mu_{n},\nu_{n}),
\end{equation*} 
so that all inequalities are equalities, and the proof is complete. 
\end{proof}

\smallskip
\noindent Note that the conditions of the above result are met if, for example $L_{n} \to L$ in the H{\"{o}}lder norm.

\smallskip
%=================================================
%=================================================
%=================================================

\section{The Special Case of Expected Shortfall}
\label{ESsection}

In this section, we examine the case of ES. This is the spectral risk measure corresponding to $\Gamma_{\sigma} = \delta_{\alpha}$ for some $\alpha \in (0,1)$, referred to as the confidence level (equivalently, $\sigma = (1-\alpha)^{-1} \mathbf{1}_{[\alpha,1]}$).  ES is an extremely popular risk measure both in academic research and practical applications, and the basis for the market risk capital charge in the current international standard for banking regulations (\citetalias{BaselMarketRiskCharge}).

\smallskip
%=================================================
\subsection{Definition and Problem Formulation}
ES has an explicit dual representation as a coherent risk measure (e.g., \citet[Theorem 4.52]{FollmerSchied}): 
\begin{equation*} 
\label{ESDualRep}
\mbox{ES}_{\alpha,\pi}(L) = \max_{\Theta\in G_{\alpha}(\pi)} \mathbb{E}_{\Theta}[L],
\end{equation*} 
where $G_{\alpha}(\pi)$ is the set of all probability measures $\Theta$ that are absolutely continuous with respect to $\pi$, with density satisfying $\tfrac{d\Theta}{d\pi} \leq (1-\alpha)^{-1}$:
\begin{equation*} 
\label{SetGalphaPi}
G_{\alpha}(\pi):= \left\{ \Theta \in \cP(\cW) \; \vert \;  \Theta\ll \pi, \frac{d\Theta}{d\pi} \leq (1-\alpha)^{-1}\right\},
\end{equation*} 
\noindent with the inequality holding $\pi$-a.s. Furthermore, it is known (e.g., \citet[Remark 4.53]{FollmerSchied}) that the above maximum is attained by the probability measure $\Theta_{0}\in G_{\alpha}(\pi)$ with density
\begin{equation} \label{OptimalThetaFormula}
\frac{d\Theta_{0}}{d\pi} = \frac{1}{1-\alpha}\left(\mathbf{1}_{\{L > q\}} + \kappa \mathbf{1}_{\{L=1\}}\right),
\end{equation} 
where $q$ is an $\alpha$-quantile of $L$, and where $\kappa$ is defined as
\begin{equation*} 
\kappa := \begin{cases} 
0 & \pi(L=q)=0, \\
\frac{(1-\alpha)-\pi(L>q)}{\pi(L=q)} & \mbox{otherwise.}
\end{cases}
\end{equation*} 

\bigskip

Just as we defined the MSP as a robust version of the spectral risk measure for a given loss random variable $L$ on $\cW$, we define here the Maximum Expected Shortfall at confidence level $\alpha$, consistent with given prescribed marginals, as the maximum value of $\mbox{ES}_{\alpha,\pi}(L)$ among all measures $\pi$ with the prescribed marginals:

\smallskip 
\begin{definition}\label{DefMES}
For a given loss random variable $L:\cW\to\mathbb{R}$ and given marginal distributions $\mu$ on $\cX$ and $\nu$ on $\cY$, the Maximum Expected Shortfall ($\mathrm{MES}$) at confidence level $\alpha$ associated with $L$ is defined as
\begin{equation*}
\mathrm{MES}_{\alpha}(L) :=\sup_{\pi\in \Pi(\mu,\nu)} \mathrm{ES}_{\alpha,\pi}(L).
\end{equation*} 
\end{definition}

\noindent It then follows that 
\begin{align*} 
\mbox{MES}_{\alpha}(L)&=\sup_{\pi\in \Pi(\mu,\nu)} \mbox{ES}_{\alpha,\pi}(L) \\
&= \sup_{\pi\in\Pi(\mu,\nu)} \min_{\beta\in\mathbb{R}}(\beta + (1-\alpha)^{-1}\mathbb{E}_{\pi}[(L-\beta)_{+}]) \\
&= \sup_{\pi\in\Pi(\mu,\nu),\Theta\in G_{\alpha}(\pi)} \mathbb{E}_{\Theta}[L].
\end{align*}

\noindent If we want to emphasize the dependence on the marginal distributions $\mu$ and $\nu$, we employ the notation
\begin{align} 
V_{\alpha}(\mu,\nu) &:= \sup_{\pi\in\Pi(\mu,\nu)} \min_{\beta\in\mathbb{R}}(\beta + (1-\alpha)^{-1}\mathbb{E}_{\pi}[(L-\beta)_{+}]) 
\label{WorstCVaRCompactPrimal} \\
&= \sup_{(\pi,\Theta)\in F_{\alpha}(\mu,\nu)} \int_{\cX \times \cY} L(x,y) d\Theta,  \label{WorstCVaRPrimal} 
\end{align}
\noindent where the correspondence $F_{\alpha}: \cP(\cX) \times \cP(\cY) \rightrightarrows \cP(\cX \times \cY) \times \cP(\cX \times \cY)$ is defined as
\begin{equation*} \label{PrimalFeasibleSetMapping}
F_{\alpha}(\mu,\nu) := \left\{ (\pi,\Theta) \; \vert \; \pi\in\Pi(\mu,\nu), \Theta\in G_{\alpha}(\pi)\right\}.
\end{equation*} 

\smallskip

As we have shown above, Problem~\eqref{WorstCVaRCompactPrimal} has a dual problem naturally associated with it: 
\begin{equation}\label{WorstCVaRDual}
D_{\alpha}(\mu,\nu) := \inf_{(\varphi,\psi,\beta) \in H_{\alpha}(L)} \int_{\cX} \varphi\, d\mu + \int_{\cY} \psi\, d\nu + \beta,
\end{equation} 
where $H_{\alpha}$ is the set of all $(\varphi,\psi,\beta)\in \cL^{1}(\mu)\times \cL^{1}(\nu)\times \bbR$ for which $(1-\alpha)(\varphi(x)+\psi(y)) \geq \max(L(x,y)-\beta,0)$. 

\bigskip\noindent
Alternatively, one could attempt to derive a dual problem from~\eqref{WorstCVaRPrimal}. Based on formal calculations similar to those used earlier, one arrives at: 
\begin{equation} \label{WorstCVaRDual}
\tilde D_{\alpha}(\mu,\nu) := \inf_{(\varphi,\psi,\rho,\beta)\in \tilde H_{\alpha}(L)} \int_{\cX} \varphi d\mu + \int_{\cY} \psi d\nu + \beta,
\end{equation} 
where $\tilde H_{\alpha}(L)$ is the set of all $(\varphi,\psi,\rho,\beta)\in C_{b}(\cX)\times C_{b}(\cY) \times C_{b}(\cX\times\cY) \times \mathbb{R}$ such that for all $(x,y)\in \cX \times \cY$:
\begin{gather*}
(1-\alpha)(\varphi(x) + \psi(y)) \geq \rho(x,y) \\
\rho(x,y) + \beta \geq L(x,y) \\
\rho(x,y) \geq 0.
\end{gather*}
This version of the duality can be proved directly by starting with Problem~\eqref{WorstCVaRDual} and applying results from convex duality together with an approximation argument, as in~\citet{VillaniTopicsInOT}. While perhaps more complicated, the pair~\eqref{WorstCVaRPrimal}, \eqref{WorstCVaRDual} is also more explicit, with the primal problem being linear in the decision variables $(\pi,\Theta)$, and the extra dual variables $\rho$ having a natural interpretation as slack variables for the constraint $(1-\alpha)(\varphi(x)+\psi(y)) \geq \max(L(x,y)-\beta,0)$.\footnote{One could attempt a similar treatment for the general spectral risk measure problem. However, this would require introducing a decision variable $\Theta_{\alpha}\in \cP(\cX\times\cY)$ for each $\alpha\in (0,1)$; and it is unclear whether any benefit would be derived from these complications.}

\bigskip
In the remainder of this section, we investigate two topics: (i) the first is attainment of the optimal solution in the dual problem $D_{\alpha}(\mu,\nu)$; and, (ii) the second is continuity of the feasible set correspondence $F_{\alpha}(\mu,\nu)$. We recall the function $g_{\alpha}(b,\pi)$ of~(\ref{gDefinition}). The point of the next lemma is to ensure the existence of an interval $[k^{*},K^{*}]$ containing the minimizers of $g(\cdot,\pi)$ for {\em all\/} $\pi\in\Pi(\mu,\nu)$, so that $\min_{b\in\mathbb{R}} g(b,\pi) = \min_{b\in [k^{*},K^{*}]} g(b,\pi)$ simultaneously for all $\pi\in\Pi(\mu,\nu)$.

\smallskip
\begin{lemma}
Let $\alpha\in (0,1)$, and suppose that {\bf A.2, A.3} hold. Then there exist $k^{*},K^{*}$ such that 
\begin{align*}
\pi(L \geq K^{*}) &< 1-\alpha, \quad \forall\pi\in\Pi(\mu,\nu), \\
\pi(L > k^{*}) &> 1-\alpha, \quad \forall\pi\in\Pi(\mu,\nu).
\end{align*}
\end{lemma}

\begin{proof} 
$\pi(L\geq K^{*}) \leq \pi(A(x)+B(y) \geq K^{*}) \leq \mu(A(x)\geq K^{*}/2) + \nu(B(y) \geq K^{*}/2)$, which is less than $1-\alpha$ for $K^{*}$ large enough, while $\pi(L > k^{*}) \geq \pi(a(x) + b(y) > k^{*}) \geq \pi(\min(a(x),b(y)) > k^{*}/2) \geq 1 - \mu(a(x) \leq k^{*}/2) - \nu(b(y) \leq k^{*}/2)$, which is greater than $1-\alpha$ for $k^{*}$ small enough.
\end{proof}

\smallskip
We note that the above lemma is also a consequence of known bounds on the VaR of a sum given the marginal distributions of its components, and that (semi)-explicit formulas for the best general values for $k^{*},K^{*}$ are known in terms of the distributions of $A(x),B(y),a(x)$ and $b(y)$ (e.g., \citet{RuschendorfVaR} or~\citet{MakarovVaR}). We now note the following.

\begin{itemize} 
\item For a fixed $\beta\in\mathbb{R}$, $g_{\alpha}(\beta,\cdot)$ is concave and upper-semicontinuous in $\pi$ (by Proposition~\ref{HPropertiesProposition} with $\Gamma_{\sigma} = \delta_{\alpha}$). 
\item For a fixed $\pi\in\Pi(\mu,\nu)$, $g_{\alpha}(\cdot,\pi)$ is convex and real-valued on $\mathbb{R}$ and therefore convex and continuous on $[k^{*},K^{*}]$.
\end{itemize}

\noindent By \citet[Theorem 2.10.2]{Zalinescu}, it follows that under {\bf A.1-A.3} the following holds:\footnote{A special case of this result was given in~\citet{MemartoluieThesis}.}
\begin{equation*}
\max_{\pi\in\Pi(\mu,\nu)} \min_{\beta\in\mathbb{R}}(\beta + (1-\alpha)^{-1}\mathbb{E}_{\pi}[(L-\beta)_{+}])  
= \min_{\beta\in\mathbb{R}}\max_{\pi\in\Pi(\mu,\nu)} (\beta + (1-\alpha)^{-1}\mathbb{E}_{\pi}[(L-\beta)_{+}]).
\end{equation*} 

\noindent
In particular, there is a value $\beta^{*} \in [k^{*},K^{*}]$ at which the minimal value for the dual problem is attained. For this $\beta^{*}$, standard results (e.g., \citet[Theorem 5.10]{VillaniOTOldAndNew}) imply that there is an optimal dual solution $(\varphi^{*},\psi^{*})\in \cL^{1}(\mu)\times \cL^{1}(\nu)$ for the optimal transport problem $\max_{\pi\in\Pi(\mu,\nu)} \mathbb{E}_{\pi}[(L-\beta^{*})_{+}]$, and therefore $(\beta^{*},\varphi^{*},\psi^{*})$ is an optimal dual solution for our problem. Attainment of the primal solution $\pi^{*} \in \Pi(\mu,\nu)$ has already been shown. The existence of an optimal $\Theta^{*}$ follows from~\eqref{OptimalThetaFormula}, or the fact that $G_{\pi^{*}}(\alpha)$ is compact.

\smallskip
%=================================================

\subsection{Continuity of the Feasible Set Correspondence}
In this section, we investigate the continuity of the feasible set correspondence $F_{\alpha}$. This can be used to derive stability results for the optimal pairs $(\pi^{*},\Theta^{*})$ in a manner similar to the treatment of the spectral risk measure problem above. Given the fact that $\Pi$ is continuous, continuity of $F_{\alpha}$ is a consequence of the following result.

\smallskip
\begin{theorem}\label{GAlphaStability}
The correspondence $G_{\alpha}:\cP(\cW) \rightrightarrows \cP(\cW)$ is continuous.
\end{theorem}

The remainder of this section is devoted to proving this result. We begin by showing that a bound on the Radon-Nikodym derivatives of $\Theta$ with respect to $\pi$ is preserved under weak convergence.

\smallskip
\begin{lemma}  \label{StabilityOfRNBound}
Let $\{\pi_{n}\}_{n\in\mathbb{N}}$ and $\{\Theta_{n}\}_{n\in\mathbb{N}}$ be sequences in $\cP(\cW)$, and suppose that $\pi_{n}\to\pi$, $\Theta_{n}\to\Theta$, and $\Theta_{n} \ll\pi_{n}$, with $\tfrac{d\Theta_{n}}{d\pi_{n}}\leq M$ ($\pi_{n}$-a.s.) for all $n \in \mathbb{N}$ and some $M\in [1,\infty)$. Then $\Theta \ll \pi$, and $\tfrac{d\Theta}{d\pi} \leq M$ ($\pi$-a.s.). 
\end{lemma}

\begin{proof} 
Let $f\in C_{b}(\cW)$, $f\geq 0$. Then: 
\begin{equation*} 
\int f d\Theta  = \lim_{n\to\infty} \int f \tfrac{d\Theta_{n}}{d\pi_{n}} d\pi_{n} 
\leq M \lim_{n\to\infty} \int f d\pi_{n} = M \int f d\pi.
\end{equation*} 
Let $U$ be an open set, and $\{f_{k}\}_{k \in \mathbb{N}}$ a sequence of continuous functions $0\leq f_{k} \leq 1$, for all $k \in\mathbb{N}$, converging pointwise to $\mathbf{1}_{U}$. Dominated convergence then implies that $\Theta(U) \leq M \pi(U)$, and the regularity of the measures $\Theta$ and $\pi$ yields $\Theta \ll \pi$, and in fact $\Theta(A) \leq M\pi(A)$ for any measurable set $A$. Let $A := \left\{ \tfrac{d\Theta}{d\pi} > M\right\}$, and assume by way of contradiction that $\pi(A) > 0$. Then $\Theta(A) = \int_{A} \tfrac{d\Theta}{d\pi} \, d\pi > M\pi(A)$, a contradiction.
\end{proof}

\begin{proposition}\label{FeasibleSetUpperHemicontinuity}
The correspondence $G_{\alpha}:\cP(\cW) \rightrightarrows \cP(\cW)$ is upper hemicontinuous.
\end{proposition} 

\begin{proof} 
Let $\pi_{n} \to\pi$, and $\Theta_{n} \in G_{\alpha}(\pi_{n})$ for each $n$. Since $\Theta_{n}(A) \leq (1-\alpha)^{-1}\pi_{n}(A)$ for all $A$, $\{\Theta_{n}\}_{n\in\mathbb{N}}$ is relatively compact. We thus obtain a subsequence $\{\Theta_{n_{k}}\}_{k\in\mathbb{N}}$ converging to some $\Theta\in\cP(\cW)$. The fact that $\Theta \in G_{\alpha}(\pi)$ follows from Lemma~\ref{StabilityOfRNBound}.
\end{proof}

We will need the following technical lemma. 

\begin{lemma}\label{TechnicalShiftLemma}
Suppose that $m\in\cP(\cW)$, $A > 1$, and $g\in \cL^{\infty}(m)$, with $0\leq g\leq A$ and $\int g\, dm = c \in (0,1)$, with $\delta^{2}-4c(1-c) > 0$, where  $\delta = A-1$. For $r \geq 0$, define $g_{r} = \min(A,g+r)$. Then there exists $r$ such that $\int g_{r}\, dm = 1$ and $r\leq k_{c,\delta}$ where: 
\begin{equation*} 
k_{c,\delta} = A - c - \frac{\delta + \sqrt{\delta^{2}-4c(1-c)}}{2}.
\end{equation*} 
\end{lemma}

\begin{proof}
For simplicity, denote $k_{c,\delta}$ by $k$, and observe that $k > 0$. Now,
\begin{align*} 
\int g_{k}\, dm &= \int (g+k)\cdot \mathbf{1}_{\{g\leq A-k\}}\, dm + \int A\cdot \mathbf{1}_{\{g > A-k\}}\, dm \\
&= c + k \cdot m(g\geq A-k) + \int (A-g)\cdot \mathbf{1}_{\{g > A-k\}}\, dm \\
& \geq c + k\cdot m(g \geq A-k).
\end{align*}
Take
\begin{equation*} 
\ve := A - c - k = \frac{\delta + \sqrt{\delta^{2}-4c(1-c)}}{2} > 0,
\end{equation*} 
so that $A-k = c+\ve$, and $k = A-(c+\ve)$. By Markov's inequality $m(g \leq c + \ve) \geq \frac{\ve}{c + \ve}$, so that
\begin{align*}
\int g_{k}\, dm &\geq c + (A-(c+\ve))\frac{\ve}{c+\ve} \\
&= \frac{c^{2} + c\ve + \ve A-\ve c-\ve^{2}}{c+\ve} \\
&=  \frac{c+\ve-(\ve^{2}-\delta \ve+c(1-c))}{c+\ve} =1.
\end{align*} 
Define $Q(r) := \int g_{r} \, dm$. Then $Q(0) = c$, $Q(k) \geq 1$, and $Q$ is continuous by Dominated Convergence. Thus, there exists $r \in (0,k]$ such that $Q(r)=1$.
\end{proof}

\smallskip \noindent
We note that for a fixed $A$ (and therefore $\delta$), $k_{c,\delta} \to 0$ as $c\to 1$.

\bigskip
\begin{lemma}\label{ContinuousRNLemma}
Suppose that $\pi\in\cP(\cW)$, $\{\pi_{n}\}_{n\in\mathbb{N}}$ is such that $\pi_{n}\to\pi$, $\Theta\in G_{\alpha}(\pi)$, and there is a version of $\frac{d\Theta}{d\pi} \in C_{b}(\cW)$. Then there exists a sequence $\{\Theta_{n}\}_{n\in\mathbb{N}}$, with $\Theta_{n}\in G_{\alpha}(\pi_{n})$, such that $\Theta_{n}\to\Theta$.
\end{lemma}

\begin{proof}
Let $\eta\in C_{b}(\cW)$ be a version of $\frac{d\Theta}{d\pi}$ with $0\leq \eta\leq (1-\alpha)^{-1}$. Let $c_{n} := \int \eta d\pi_{n}$, for each $n$. Since $\eta\in C_{b}(\cW)$, $c_{n} \to 1$. If $c_{n} \geq 1$, define $d\Theta_{n} = c_{n}^{-1}\eta d\pi_{n}$, and note that for such $n$ and $f\in C_{b}(\cW)$: 
\begin{gather} 
\left| \int f d\Theta_{n} - \int f d\Theta\right| = \left| \int f\eta c_{n}^{-1} d\pi_{n} - \int f\eta d\pi\right| 
\leq (1-c_{n}^{-1}) \|f\eta\|_{\infty} +  \left| \int f\eta d\pi_{n} - \int f\eta d\pi\right| \to 0. \nonumber
\end{gather}
If $c_{n} < 1$, apply Lemma~\ref{TechnicalShiftLemma} with $c=c_{n}$, $A=(1-\alpha)^{-1}$, $\delta = (1-\alpha)^{-1}-1$,
and $g=\eta$, to obtain $g_{r_{n}} = \min((1-\alpha)^{-1},\eta + r_{n})$ with $\int g_{r_{n}}\, d\pi_{n}=1$ and $r_{n} \leq k_{c_{n},\delta}$, so that $r_{n}  \to 0$ as $n\to\infty$. Then, with $d\Theta_{n} = g_{r_{n}}d\pi_{n}$,
\begin{gather} 
\left| \int f d\Theta_{n} - \int f d\Theta\right| = \left| \int f g_{r_{n}} d\pi_{n} - \int f\eta d\pi\right| \leq \|f\|_{\infty} r_{n} +  \left| \int f\eta d\pi_{n} - \int f\eta d\pi\right| \to 0. \nonumber
\end{gather}
\end{proof}

\smallskip
\begin{proposition}\label{FeasibleSetLowerHemicontinuity}
The correspondence $G_{\alpha}:\cP(\cW) \rightrightarrows \cP(\cW)$ is lower hemicontinuous.
\end{proposition} 

\begin{proof} 
Consider a sequence $\{\pi_{n}\}_{n\in\mathbb{N}}$ with $\pi_{n}\to\pi$. We will show the existence of a sequence $\{\Theta_{n}\}_{n\in\mathbb{N}}$, with $\Theta_{n}\in G_{\alpha}(\pi_{n})$ such that $\Theta_{n} \to\Theta$.

\smallskip
Let $\xi = \tfrac{d\Theta}{d\pi}$. If $\xi$ has a continuous version, the result follows from Lemma~\ref{ContinuousRNLemma}. Otherwise, by~\citet[Section 37, Theorem 2]{KolmogorovFomin}, there exist continuous $\xi_{m}$, for each $m\in\mathbb{N}$, such that $\xi_{m} \to \xi$ in $\cL^{1}(\pi)$. We can clearly take $0\leq \xi_{m}\leq (1-\alpha)^{-1}$, for each $m \in\mathbb{N}$. Let $c_{m} := \int \xi_{m}\, d\pi$, for each $m \in\mathbb{N}$, and note that $c_{m} \to 1$.  

\smallskip
For a given $m \in\mathbb{N}$, if $c_{m} \geq 1$, define $\xi'_{m} := c_{m}^{-1}$, and we have $\|\xi'_{m}-\xi\|_{1} \leq (1-c_{m}^{-1})\|\xi_{m}\|_{1} + \|\xi_{m}-\xi\|_{1}$. If $c_{m} < 1$, apply Lemma~\ref{TechnicalShiftLemma} to obtain $r_{m} \leq k_{c_{m},\delta}$ (with $\delta := (1-\alpha)^{-1}-1$) such that $\xi'_{m} = \min((1-\alpha)^{-1},\xi_{m}+r_{m})$, with $\int \xi'_{m}\, d\pi =1$ and $\|\xi'_{m} - \xi\|_{1} \leq r_{m} + \|\xi_{m}-\xi\|_{1}$. Either way, $\xi'_{m}\to\xi$ in $\cL^{1}(\pi)$, and defining $\eta_{m}$ by $d\eta_{m} := \xi'_{m}d\pi$ gives $\eta_{m} \in G_{\alpha}(\pi), \eta_{m}\to\Theta$.

\smallskip
For each $n \in\mathbb{N}$, the set $G_{\alpha}(\pi_{n})$ is compact, and therefore there exists $\Theta_{n} \in G_{\alpha}(\pi_{n})$ such that 
\begin{equation*} 
d_{\cP}(\Theta_{n},\Theta) = \inf\Big\{ d_{\cP}(m,\Theta) :  m\in G_{\alpha}(\pi_{n})\Big\}.
\end{equation*} 

\smallskip\noindent
Let $\{\Theta_{n_{k}}\}_{k\in\mathbb{N}}$ be a subsequence of $\{\Theta_{n}\}_{n\in\mathbb{N}}$. By Lemma~\ref{ContinuousRNLemma}, for each $m \geq 1$, there exists a sequence $\{\Theta_{m,n_{k}}\}_{k\in\mathbb{N}}$ with $\Theta_{m,n_{k}} \in G_{\alpha}(\pi_{n_{k}})$ and $\Theta_{m,n_{k}} \to \eta_{m}$. For each $j\in\mathbb{N}$, there exists $M(j)>M(j-1)$ (with $M(0)=0$) such that $d_{\cP}(\eta_{m},\Theta) \leq \tfrac{1}{2j}$ for $m\geq M(j)$, and an $N(j)>N(j-1)$ (with $N(0)=0$) such that $d_{\cP}(\Theta_{M(j),n_{k}},\eta_{M(j)}) \leq \tfrac{1}{2j}$ for all $n_k\geq N(j)$. The sequence $\{\Theta_{M(j),N(j)}\}_{j\in\mathbb{N}}$ satisfies  $\Theta_{M(j),N(j)} \in G_{\alpha}(\pi_{N(j)})$, for each $j\in\mathbb{N}$, and $d_{\cP}(\Theta_{M(j),N(j)},\Theta) \leq d_{\cP}(\Theta_{M(j),N(j)},\eta_{M(j)}) + d_{\cP}(\eta_{M(j)},\Theta) \leq \tfrac{1}{j}$, so that $\Theta_{M(j),N(j)}\to \Theta$. But $\{N(j)\}_{j\in\mathbb{N}}$ is a subsequence of $\{n_{k}\}_{k\in\mathbb{N}}$ and by definition, we have 
\begin{equation*} 
d_{\cP}(\Theta_{N(j)},\Theta) \leq d_{\cP}(\Theta_{M(j),N(j)},\Theta) \leq \frac{1}{j}.
\end{equation*} 
Thus, every subsequence of $\{\Theta_{n}\}_{n\in\mathbb{N}}$ has a further subsequence that converges to $\Theta$, and so $\Theta_{n}\to\Theta$.
\end{proof}

%\smallskip
%=================================================
%=================================================
%=================================================

\section{Asymptotic Distributions and Numerical Results}\label{AsymptoticDistributionSection}

In this section, we present an asymptotic result that is valid when both $N_{\cX}=|\cX|$ and $N_{\cY}=|\cY|$ are finite; and we discuss potential extensions to the case when $X$ and $Y$ do not have finite support.  Without loss of generality, we identify $\cX$ with $\{1,2,\ldots,N_{\cX}\}$ and $\cY$ with $\{1,2,\ldots,N_{\cY}\}$. The marginal probability distributions $\mu$ and $\nu$ can then be identified with probability vectors in $\bbR^{N_{\cX}}$ and $\bbR^{N_{\cY}}$, respectively, and we will do so throughout this section. 

\bigskip

For a given $\alpha\in (0,1)$, the primal maximum CVaR problem becomes the following linear program (see \citet{MemartoluieSaundersWirjanto}):
\begingroup
\allowdisplaybreaks
\begin{gather} 
\max_{\pi\in\mathbb{R}^{N_{\cX}}\times N_{\cY},\Theta\in\mathbb{R}^{N_{\cX}\times N_{\cY}}} \label{PrimalLP}
\sum_{i=1}^{N_{\cX}}\sum_{j=1}^{N_{\cY}} L_{ij}\Theta_{ij} \\
\sum_{j=1}^{N_{\cY}} \pi_{ij} = \mu_{i},\quad i=1,\ldots,N_{\cX} \nonumber \\
\sum_{i=1}^{N_{\cX}} \pi_{ij} = \nu_{j},\quad j=1,\ldots,N_{\cY} \nonumber \\
\Theta_{ij} \leq (1-\alpha)^{-1}\pi_{ij}, \quad i=1,\ldots,N_{\cX},\; j=1,\ldots,N_{\cY} \nonumber \\
\sum_{i=1}^{N_{\cX}}\sum_{j=1}^{N_{\cY}} \Theta_{ij} = 1 \nonumber \\
\pi_{ij},\Theta_{ij} \geq 0, \quad i=1,\ldots,N_{\cX},\; j=1,\ldots,N_{\cY} \nonumber 
\end{gather}
\endgroup

\noindent The feasible set is nonempty ($\Theta_{ij} = \pi_{ij} = p_{i}q_{j}$ is feasible), closed, and bounded; and the objective function is continuous. Therefore, the primal problem has a finite value and an optimal solution that attains that value.

\vspace{0.2cm}

The dual of the above linear program is
\begin{gather*} 
\min_{\varphi\in\mathbb{R}^{N_{\cX}},\psi\in\mathbb{R}^{N_{\cY}},\rho\in\mathbb{R}^{N_{\cX}\times N_{\cY}},\beta\in\mathbb{R}} \label{DualLP}
\sum_{i=1}^{N_{\cX}}\varphi_{i} \mu_{i} + \sum_{j=1}^{N_{\cY}} \psi_{j} \nu_{j} + \beta \\
(1-\alpha)(\varphi_{i} + \psi_{j}) - \rho_{ij} \geq 0, \quad i=1,\ldots,N_{\cX},\; j=1,\ldots,N_{\cY} \nonumber \\
\rho_{ij} + \beta \geq L_{ij}, \quad i=1,\ldots,N_{\cX},\; j=1,\ldots,N_{\cY} \nonumber \\
\rho_{ij} \geq 0, \quad i=1,\ldots,N_{\cX},\; j=1,\ldots,N_{\cY} \nonumber 
\end{gather*}

\bigskip

Note that for any feasible dual solution $(\varphi,\psi,\rho,\beta)$ and $c\in\mathbb{R}$, $(\varphi+c,\psi-c,\rho,\beta)$ is also feasible, with the same objective value. Consequently, we can add the constraint $\varphi_{1} = 0$, without affecting the optimal value of the dual. With this extra constraint, there are finitely many extreme points $(\varphi^{k},\psi^{k},\rho^{k},\beta^{k})$, $k=1,\ldots,N_{D}$ of the dual feasible polyhedron (notice that these do not depend on $\mu,\nu$). For probability vectors, $\mu\in \mathbb{R}^{N_{\cX}}, \nu\in \mathbb{R}^{N_{\cY}}$, we then have by linear programming duality: 
\begin{equation} 
V_{\alpha}(\mu,\nu) = \min_{k=1,\ldots,N_{D}} \varphi^{k}\cdot \mu + \psi^{k}\cdot \nu + \beta^{k} 
\label{WDDef}.
\end{equation}

\smallskip
If independent random samples of size $n$ are generated from $\mu$ and $\nu$, then the empirical distributions $\mu_{n}$ and $\nu_{n}$ are random probability vectors and the Central Limit Theorem (e.g., \citet[p.\ 16]{vanderVaart}) implies that
\begin{equation*} 
\sqrt{n}\left(\left(\begin{matrix} \mu_{n} \\ \nu_{n} \end{matrix}\right) - 
\left( \begin{matrix} \mu \\ \nu \end{matrix}\right)\right) \rightsquigarrow N(0,\Sigma),
\end{equation*} 
where $\rightsquigarrow$ denotes convergence in distribution, and the covariance matrix $\Sigma$ has the block form: 
\begin{equation*} 
\Sigma = \left(\begin{matrix} \Sigma_{\cX} & 0 \\ 0 & \Sigma_{\cY} \end{matrix}\right),
\end{equation*} 
with: 
\begin{equation} 
\Sigma_{\cX}(i,j) := \begin{cases} 
\mu_{i}(1-\mu_{i}) & i=j \\
-\mu_{i}\mu_{j} & i\ne j
\end{cases},
\qquad 
\Sigma_{\cY}(i,j) := \begin{cases} 
\nu_{i}(1-\nu_{i}) & i=j \\
-\nu_{i}\nu_{j} & i\ne j.
\end{cases} 
\label{CovarianceMatricesDef}
\end{equation} 

\smallskip
Then, arguing as in~\citet{SommerfeldMunk} (or directly verifying the conditions in Proposition 3.5 of~\citet{KlattMunkZemel}) we arrive at the following.

\smallskip
\begin{theorem}\label{LPCLT}
Suppose that $\alpha\in (0,1)$, $|\cX|=N_{\cX} < \infty$ and $|\cY| = N_{\cY} < \infty$. Let $K$ be the set of all $(\mu,\nu)$ with $\mu$ and $\nu$ probability vectors in $\mathbb{R}^{N_{\cX}}$ and $\mathbb{R}^{N_{\cY}}$ respectively. Then: 

\begin{enumerate}
\item $V_{\alpha}(\cdot,\cdot)$ is Hadamard directionally differentiable, tangentially to $K$, with derivative: 
\begin{equation} 
V'_{\alpha}(\mu,\nu;d_{\mu},d_{\nu}) = \min_{k\in I(\mu,\nu)} \varphi^{k}\cdot d_{\mu} + \psi^{k}\cdot d_{\nu},
\label{ValueFunctionDerivative}
\end{equation}
where $I(\mu,\nu)$ is the set of all $k$ which attain the minimum in~\eqref{WDDef}. 
\smallskip
\item Suppose that $\mu_{n},\nu_{n}$, $n=1,2,\ldots$ are the empirical measures from independent random samples from 
$\mu$ and $\nu$ respectively. Then: 
\begin{equation*} 
\sqrt{n}(V_{\alpha}(\mu_{n},\nu_{n})-V_{\alpha}(\mu,\nu)) \rightsquigarrow 
V'_{\alpha}(\mu,\nu;Z_{\cX},Z_{\cY}),
\end{equation*} 
where $Z_{\cX}$, $Z_{\cY}$ are independent centred normal random vectors with respective covariance matrices $\Sigma_{\cX}$ and $\Sigma_{\cY}$ as in~\eqref{CovarianceMatricesDef}.
\end{enumerate}
\end{theorem}

\smallskip
The limiting distribution in the previous theorem will be Gaussian when~\eqref{ValueFunctionDerivative} is linear in $(d_{\mu},d_{\nu})$. This will be true when the optimal solution to the dual problem (with the additional constraint $\varphi_{1} = 0$) is unique. With multiple optimal dual solutions, the limiting distribution will be non-Gaussian. This is directly analogous to the results for the optimal transport problem with $|\cX|$ and $|\cY|$ finite in~\citet{SommerfeldMunk}. For the quadratic cost function with $\cX$ and $\cY$ subsets of $\mathbb{R}^{n}$, and under technical conditions on the measure $\nu$, \citet{DelBarrioLoubes} show that there is a unique dual solution to the optimal transport problem, and that a central limit theorem for the value function holds. It is an interesting direction for future work to attempt to generalize this result to other cost functions, and to our risk measure bounding problem. 

\bigskip

\begin{example}[Linear Loss with Gaussian Marginals]
As a first example, consider the loss $L(X,Y) = X+Y$, with $X\sim N(0,1)$ and $Y\sim N(0,1)$. The optimal coupling is $X=Y$, so that $L$ is Gaussian with  mean 0 and standard deviation 2. For generated random samples, the MES will again be attained by a comonotonic coupling. The ES of the sampled problem will then simply be the sum of the estimated ESs from the samples from $X$ and $Y$. Known results on the estimation of risk measures (e.g., \citet{ManistreHancock}) imply that the limit distribution will be Gaussian. Figures~\ref{NormalHistogram} and~\ref{NormalQQPlot} show the results of a computational experiment, in which a sample of size 200 was generated from $X$, a sample of size 400 was simulated from $Y$, and the MES with $\alpha = 0.9$ was estimated by solving the linear program~\eqref{PrimalLP}. This experiment was repeated 1000 times, and the histogram of the optimal values is presented in Figure~\ref{NormalHistogram}. The q-q plot against a fitted normal is displayed in Figure~\ref{NormalQQPlot}. An Anderson-Darling test of normality was not able to reject the null hypothesis of a normal distribution for the simulated optimal values at the 95\% confidence level.

\begin{figure}[h!]
\begin{center}
\subfloat[\footnotesize{Histogram of the optimal value. The red line indicates a fit to a normal distribution.\vspace{-0.4cm}}]{\label{NormalHistogram}
\includegraphics[width=8.6cm, trim={0.5cm 0.1cm 0.5cm 0cm}]{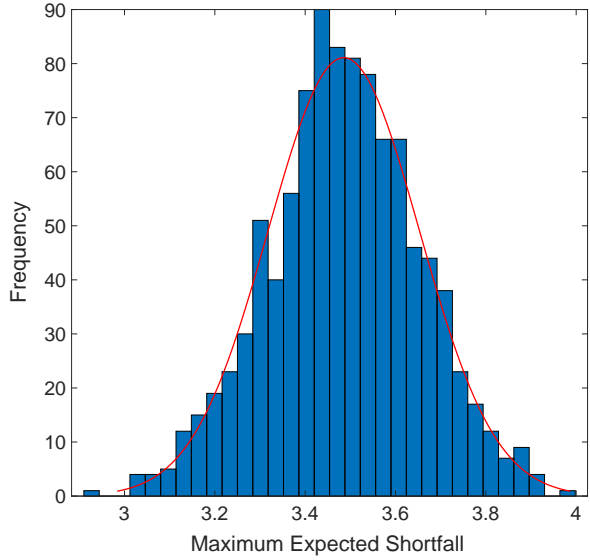}
}
\subfloat[\footnotesize{Quantiles of the simulated optimal value.}]{\label{NormalQQPlot}
\includegraphics[width=8.6cm, trim={1.5cm 0cm 0.5cm 0cm}]{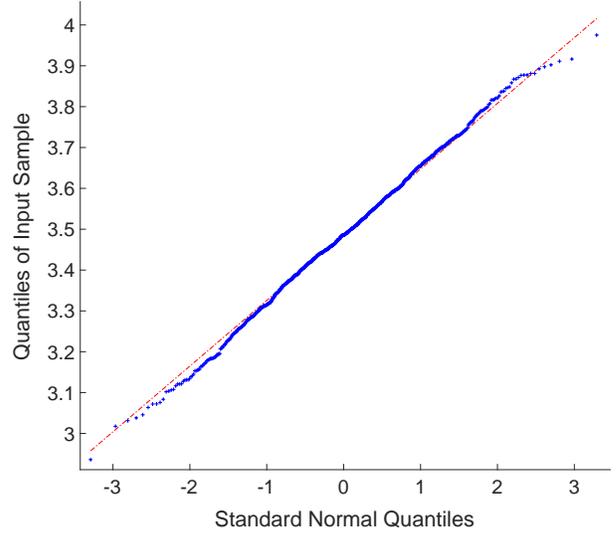}
}
\end{center}
\caption{{\footnotesize{Histogram of the optimal value of the MES (at $\alpha = 0.9$), with $X$ and $Y$ simulated from the standard Gaussian distribution, and losses given by $L = X+Y$; as well as quantiles of the simulated optimal value of the MES (at $\alpha = 0.9$), with $X$ and $Y$ simulated from the standard Gaussian distribution, and losses given by $L = X+Y$.\vspace{0.8cm}}}}
\label{fig:1}
\end{figure}
\end{example}

\smallskip
\begin{example}[Counterparty Credit Risk]
A loss function related to counterparty credit risk (see, e.g.~\citet{MemartoluieSaundersWirjanto}) would be the following: 
\begin{equation*} 
L(X,Y) = \max(Y_{1},0)\cdot \Phi\left(\frac{\Phi^{-1}(PD_{1}) - \sqrt{\rho_{1}}X}{\sqrt{1-\rho_{1}}}\right)
+ \max(Y_{2},0)\cdot \Phi\left(\frac{\Phi^{-1}(PD_{2}) - \sqrt{\rho_{2}}X}{\sqrt{1-\rho_{2}}}\right),
\end{equation*} 
where $\Phi$ is the standard normal cumulative distribution function. This gives the systematic credit losses in the Vasicek model, with systematic credit factor $X\sim N(0,1)$, for a portfolio consisting of two counterparties  with probabilities of default $PD_{1}$ and $PD_{2}$, systematic credit factor loadings $\rho_{1}$ and $\rho_{2}$, and counterparty portfolio values $Y_{1}$ and $Y_{2}$, where we assume for simplicity that
\begin{equation*} 
\left(\begin{matrix}
Y_{1} \\
Y_{2} \end{matrix}\right) \sim N\left(\left(\begin{matrix} \mu_{1}\\ \mu_{2}\end{matrix}\right), 
\left(\begin{matrix} \sigma_{1}^{2} & r\sigma_{1}\sigma_{2} \\ r \sigma_{1}\sigma_{2} & \sigma_{2}^{2}\end{matrix}\right)\right).
\end{equation*} 
We simulated 500 values from each of $X$ and $Y$, and then solved problem~\eqref{PrimalLP} for the MES at $\alpha = 0.9$, with the model parameters $PD_{1} = PD_{2} = 0.02$, $\rho_{1} = \rho_{2} = 0.2$, $r=0.5$, $\mu_{1} = 100$, $\mu_{2} = -100$, $\sigma_{1} = \sigma_{2} = 100$. The histogram of 1000 realized optimal values is given in Figure~\ref{GEVHistogram}. The distribution does not appear to be normally distributed, and indeed we provide on the figure the fitted Generalized Extreme Value (GEV) distribution, and the corresponding qq-plot in Figure~\ref{GEVQQPlot}. It should be noted that Theorem~\ref{LPCLT} does not apply, as we are simulated from continuous 
random variables.\footnote{Nonetheless, there is some evidence that similar properties to those identified in the Theorem may be at play here. Inspecting the optimal dual variables in the simulated linear programs showed them often to be very different in different simulation runs, which may be evidence that the optimal dual solution of the limiting problem is not unique.}

\begin{figure}[h!]
\begin{center}
\subfloat[\footnotesize{Histogram of the optimal value. The red line indicates a fit to a GEV distribution.\vspace{-0.4cm}}]{\label{GEVHistogram}
\includegraphics[width=8.5cm, trim={1cm 0.1cm 1cm 1cm}]{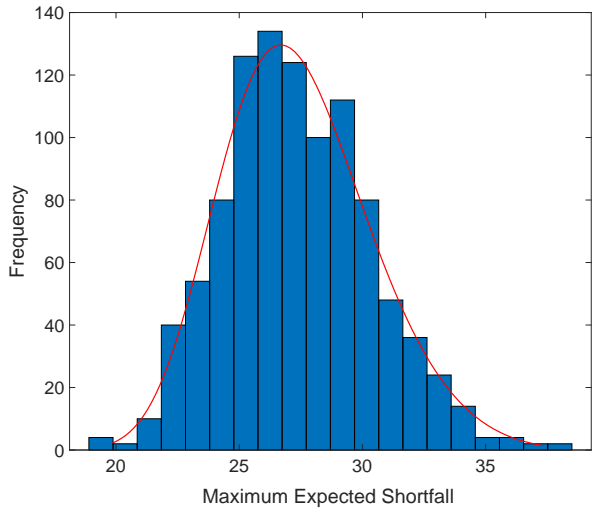}
}
\subfloat[\footnotesize{Quantiles of the simulated optimal value.}]{\label{GEVQQPlot}
\includegraphics[width=8.5cm, trim={2.5cm 0cm 1.5cm 0cm}]{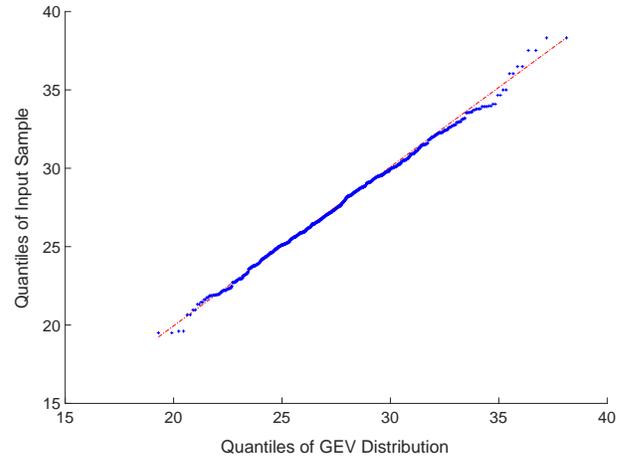}
}
\end{center}
\caption{{\footnotesize{Histogram of the optimal value of the MES ($\alpha = 0.9$), for the 
counterparty credit risk loss distribution, and quantiles of the simulated optimal value of the MES ($\alpha = 0.9$), for the 
counterparty credit risk loss distribution.}}}
\label{fig:2}
\end{figure}
\end{example}

\smallskip
%\bigskip
%=================================================
%=================================================
%=================================================

\section{Conclusions and Future Directions} 
\label{Conclusion}

In this paper, we study the problem of bounding a spectral risk measure applied to a loss function $L(X,Y)$ when the marginal distributions of the factors $X,Y$ are known, but their joint distribution is unknown. Basic properties of the optimization problem, as well as an analogue of the Kantorovich (strong) duality from optimal transport were derived. Further, we studied continuity properties of the objective value and set of maximizers with respect to perturbation of the marginal constraints, as well as the distribution of the optimal value when the factor marginal distributions are simulated from finite probability spaces.

\bigskip

There are a number of possible directions for continuation of the research in this paper, including the following: 

\begin{itemize} 
\item Bounds based solely on knowledge of marginal distributions are known to often produce extreme values in practice. Following the work of \citet{GlassermanYang} on CVA, functions that penalized a measure of ``unfitness" of $\pi$ could be introduced into the primal problem.
\smallskip
\item The distribution and convergence properties of the optimal value function could be investigated under more general assumptions on the underlying spaces $\cX$, $\cY$ than the rather restrictive assumptions imposed in Section~\ref{AsymptoticDistributionSection}. Given that this problem is not entirely resolved even in the case of optimal transport, this is likely to be challenging.
\smallskip
\item Numerical methods for the resulting optimization problems could be studied. In the work~\citet{MemartoluieSaundersWirjanto} that motivated this study, the market and credit factors were simulated independently from their marginal distributions, as in Section~\ref{AsymptoticDistributionSection}. Is there a better way to simulate the factors from their marginal distributions to yield more efficient numerical procedures? Can importance sampling algorithms be developed to yield better estimates of risk bounds given finite computational resources? Can the special structure of the optimization problem be exploited to yield more efficient numerical algorithms, given that the linear programs that must be solved can become very large?
\end{itemize}

\smallskip 
 
\noindent We leave these important questions for future research.

\newpage
%\bigskip
%=================================================
%=================================================
%=================================================

\bibliography{References}

\end{document}